\documentclass[12pt, a4paper]{amsart}
\usepackage[utf8]{inputenc}
\usepackage[T1]{fontenc}
\usepackage{lmodern}
\usepackage{microtype}
\usepackage{amsmath, amsthm, amssymb, mathtools}
\numberwithin{equation}{section}
\usepackage{braket}          
\usepackage{bbm}             
\usepackage{hyperref}
\usepackage{geometry}
\usepackage{parskip}
\usepackage{booktabs}
\usepackage{xcolor}
\usepackage{graphicx}
\usepackage{subcaption}
\usepackage{biblatex}
\usepackage[nameinlink, capitalise]{cleveref}
\usepackage{amsaddr} 
\theoremstyle{plain}
\newtheorem{theorem}{Theorem}[section]
\newtheorem{lemma}[theorem]{Lemma}
\newtheorem{corollary}[theorem]{Corollary}
\newtheorem{proposition}[theorem]{Proposition}
\theoremstyle{definition}
\newtheorem{definition}[theorem]{Definition}
\newtheorem{example}[theorem]{Example}
\theoremstyle{remark}
\newtheorem{remark}[theorem]{Remark}
\crefname{theorem}{Theorem}{Theorems}
\Crefname{theorem}{Theorem}{Theorems}
\crefname{lemma}{Lemma}{Lemmas}
\Crefname{lemma}{Lemma}{Lemmas}
\crefname{corollary}{Corollary}{Corollaries}
\Crefname{corollary}{Corollary}{Corollaries}
\crefname{proposition}{Proposition}{Propositions}
\Crefname{proposition}{Proposition}{Propositions}
\crefname{definition}{Definition}{Definitions}
\Crefname{definition}{Definition}{Definitions}
\crefname{remark}{Remark}{Remarks}
\Crefname{remark}{Remark}{Remarks}
\crefname{example}{Example}{Examples}
\Crefname{example}{Example}{Examples}
\newcommand{\C}{\mathbb{C}}

\newcommand{\R}{\mathbb{R}}

\newcommand{\W}{\mathcal{W}}
\newcommand{\Hs}{\mathcal{H}}
\newcommand{\Dfr}{\mathfrak{D}}   
\newcommand{\Kf}{K_{\varphi}} 

\newcommand{\Id}[1]{I_{#1}}
\newcommand{\rk}{\operatorname{rank}}
\newcommand{\sr}{\operatorname{sep\text{-}rank}}
\newcommand{\spr}{\operatorname{p\text{-}sep\text{-}rank}}
\newcommand{\Tr}{\operatorname{Tr}}
\newcommand{\Span}{\operatorname{span}}
\newcommand{\OSR}{\operatorname{osr}}
\newcommand{\diag}{\operatorname{diag}}
\newcommand{\HOSR}{\operatorname{hosr}}
\newcommand{\conv}{\operatorname{conv}}

\newcommand{\vect}[1]{{\boldsymbol{#1}}}
\newcommand{\Disk}{\overline{\mathbb{D}}}
\newcommand{\oDisk}{\mathbb{D}}
\newcommand{\bD}{\partial\mathbb{D}}
\newcommand{\lmax}{\lambda_{\max}}
\newcommand{\lmin}{\lambda_{\min}}
\newcommand{\norm}[1]{\left\|#1\right\|}
\newcommand{\ketbra}[2]{\ket{#1}\!\bra{#2}}

\DeclareMathOperator{\ran}{ran}

\newcommand{\DT}{\mathcal{D}_T}
\newcommand{\DTs}{\mathcal{D}_{T^\dagger }}

\title{Separable decompositions of $ {2\otimes n}$ states with operator Schmidt rank three}
\author{Perrine Vantalon \lowercase{and}  Nicolas Macris}
\address{\upshape{School of Computer and Communication Sciences - EPFL}
}

\begin{document}

\begin{abstract}
We prove that every bipartite state $\rho$ of a qubit coupled to an $n$-level system whose operator Schmidt rank equals three can be written as a mixture with $\rk (\rho)$ pure product states, and as a mixture of at most $n+1$ mixed product states in general. The second result is tight in the sense that there exist states that cannot be decomposed with a smaller number of terms. Our proof gives constructive methods to obtain the decompositions. It involves eigendecompositions of suitable unitary dilations of an operator explicitly determined by the state. 
Our framework recovers previously known results for the length of  decompositions into pure product states, and is also able to deal with decompositions into mixed states. For pure product states, our decomposition is different from existing ones. 
\end{abstract}
\maketitle

\section{Introduction}
A state (density matrix) $\rho$ of a bipartite system $\Hs_A\otimes\Hs_B$ is \emph{separable} if it is a convex combination of product states, $\rho=\sum_k p_k\,\alpha_k\otimes\beta_k$ (with $p_k\ge 0$, $\sum_k p_k=1$, and $\alpha_k,\beta_k$ states on $\Hs_A,\Hs_B$), and \emph{entangled} otherwise. Deciding separability is NP-hard in general \cite{Gurvits2003, Gharibian2010}. Finding an explicit decomposition is at least as hard, since a decomposition already certifies separability, and the same holds for decompositions of minimal length (with fewest terms).
For Hilbert spaces of dimensions $d_A, d_B$, that is $\Hs_A=\C^{d_A}$ and $\Hs_B=\C^{d_B}$, Carath\'eodory's theorem \cite{Caratheodory1911} shows that every separable state is a mixture of at most $(d_Ad_B)^2$ pure product states. The most widely used test, Peres' positive partial transpose (PPT) criterion \cite{Peres1996}, is necessary in every dimension and sufficient for $2\otimes2$ and $2\otimes3$ systems \cite{Horodecki1996}. Even in these low dimensional cases, it confirms separability without producing a decomposition.

A natural complexity measure of a bipartite operator is its \emph{operator Schmidt rank} $\OSR(\rho)$, which is the minimum number of terms in a tensor decomposition $\rho=\sum_k A_k\otimes B_k$ with unconstrained factors (in particular, no positivity or Hermiticity requirement). For separable states, one also considers the \emph{separable rank} $\sr(\rho)$ and \emph{pure separable rank} $\spr(\rho)$, defined as the minimum number of terms when the factors are required to be density matrices, and pure states respectively. One always has $\OSR(\rho)\le\sr(\rho)\le\spr(\rho)$. These inequalities hold because each rank is minimized over a class of factors that is successively more constrained (unconstrained, density matrices, pure states). 

Product states are characterized by $\OSR = 1$. States with $\OSR=2$ are separable and have separable rank two \cite{cariello2014, Johnston2014OSD, De_las_Cuevas_2019}. For a qubit coupled to a qudit, Cariello \cite{cariello2014} showed, using the construction of \cite{Kraus2000}, that operator Schmidt rank at most three implies separability. Our focus is on a new explicit construction of decompositions into pure and mixed product states on $\Hs=\C^2\otimes\C^n$.

For $n=2$, i.e., a system of two qubits, the question is sharper. Wootters' concurrence formula yields a four pure product state decomposition for every separable state \cite{Wootters1998}. Sanpera, Tarrach and Vidal obtained that the length of the minimal pure state decomposition is $\max(\rk(\rho), \rk(\rho^{T_A}))$ \cite{SanperaTarrachVidal1998}.  Jevtic, Pusey, Jennings, and Rudolph \cite{Jevtic2014} proved that the state is separable if and only if the steering ellipsoid can be enclosed in a tetrahedron contained in the Bloch sphere. Separability thus becomes the containment of one convex body in another. This method allows them to prove that separable rank is the operator Schmidt rank. 

 We generalize this geometric approach to the $2 \otimes n$ case with $\OSR = 3$. Our approach begins with a reduction, by SLOCC filtering \cite{Verstraete2001}, in order to bring the states into a normal form such that the marginal matrix $\rho_B$ (or reduced density matrix) is maximally mixed, and all the correlations between the $A$ and $B$ parties is encoded in a single $n\times n$ complex matrix $T$. The analogue of the steering ellipsoid is the numerical range of $T$ (which for $n>2$ has a more general shape). The positivity of the density matrix holds if and only if $T$ is a contraction, and the separable decompositions of the state correspond to the eigendecomposition of normal dilations. The numerical range of $T$ is enclosed by convex polygons whose vertices lie in the closed unit disk. Two dilation theorems turn this into a constructive method: the Halmos unitary dilation \cite{Halmos1950} yields a decomposition into $\rk(\rho)$ pure states, and the Wu dilation \cite[Theorem~2.2]{Wu1997} yields one with $n+1$ mixed states.

 The paper is organized as follows. Section \ref{sec:prelim} sets precise definitions of ranks and notations used throughout. The main theorems and their corollaries are stated in section \ref{sec:main}. Sections \ref{sec:setup} and \ref{sec:jnr} contain preliminary material and explain useful tools then used for the proofs of main results in section \ref{sec:decomp}. Section \ref{sec:pseudocode} summarizes the main steps of the decompositions into pure product and mixed states. Section \ref{sec:n2} discusses interesting connections between our work and topics in operator theory, geometry, and quantum channels.

\bigskip
\section{Notation and definitions}\label{sec:prelim}
\bigskip
\subsection{General notations}

Throughout, $\Hs_A=\C^{d_A}$ and $\Hs_B=\C^{d_B}$ are finite-dimensional Hilbert spaces and $\Hs=\Hs_A\otimes\Hs_B$. We write $M_d(\C)$ for the $d\times d$ complex matrices, $\Id{d}\in M_d(\C)$ for the identity, $\Tr$ for the trace, and $M^\dagger$, $M^T$, $\overline{M}$ for the adjoint, transpose, and entry-wise complex conjugate of $M$ (so $M^\dagger=(\overline{M})^T$). For Hermitian $M$ we write $M\succeq 0$ ($M\succ 0$) when $M$ is positive semi-definite (definite), and $\lmin(M)$, $\lmax(M)$ for its smallest and largest eigenvalues. The \emph{spectral norm} $\norm{M}$ is the largest singular value of $M$; $M$ is a \emph{contraction} if $\norm{M}\le 1$. A \emph{density matrix} (or \emph{state}) is a Hermitian, positive semi-definite matrix of unit trace; it is \emph{pure} if it has rank one. We also use the abbreviations PT for partial transpose, PSD for positive semi-definite. 
We denote by $\Disk=\{w\in\C : |w|\le 1\}$ the closed unit disk, by $\oDisk=\{w\in\C : |w|<1\}$ the open unit disk, and by $\bD$ the unit circle. For a set $S\subseteq\R^2$, $\conv(S)$ denotes its convex hull. The partial transpose of $\rho\in M_{d_A}(\C)\otimes M_{d_B}(\C)$ on the first factor is $\rho^{T_A}=((\cdot)^T\otimes\mathrm{id})(\rho)$, taken in a fixed product basis. For a density matrix $\rho$ over $\mathcal{H}_A\otimes \mathcal{H}_B$ the reduced density matrices $\rho_A$ and $\rho_B$ are often called $A$ and $B$-marginals. 

\subsection{Ranks}\label{sec:ranks}

\begin{definition}[Operator Schmidt ranks]
The \emph{operator Schmidt rank} of $\rho\in M_{d_A}(\C)\otimes M_{d_B}(\C)$, denoted $\OSR(\rho)$, is the minimal $r$ such that
\[
  \rho=\sum_{k=1}^r A_k\otimes B_k,
  \qquad A_k\in M_{d_A}(\C),\; B_k\in M_{d_B}(\C),
\]
with no constraint on the factors. The \emph{Hermitian operator Schmidt rank} $\HOSR(\rho)$, defined for Hermitian $\rho$, is the minimal such $r$ when the $A_k,B_k$ are in addition required to be Hermitian.
\end{definition}

\begin{definition}[Separable ranks]
Let $\rho$ be a separable density matrix on $\Hs_A\otimes\Hs_B$. Its \emph{separable rank} $\sr(\rho)$ is the minimal $r$ such that
\[
  \rho=\sum_{k=1}^r p_k\,\alpha_k\otimes\beta_k,
  \qquad p_k\geq0,\;\sum_{k=1}^r p_k=1,
\]
where the $\alpha_k\in M_{d_A}(\C)$ and $\beta_k\in M_{d_B}(\C)$ are density matrices. The \emph{pure separable rank} $\spr(\rho)$ is the minimal such $r$ when the $\alpha_k,\beta_k$ are in addition required to be pure. The pure separable rank is also called \emph{length} of the decomposition in the literature. 

\end{definition}

\begin{remark}
\label{rem:rankchain}
For bipartite Hermitian $\rho$ one has $\OSR(\rho)=\HOSR(\rho)$ \cite[Lemma 14]{De_las_Cuevas_2019}, and minimal decompositions of both kinds can be constructed from a singular value decomposition. Consequently, for separable $\rho$,
\[
  \OSR(\rho)=\HOSR(\rho)\le\sr(\rho)\le\spr(\rho).
\]
Each inequality is due to an additional constraint (positivity and purity) on the terms in the decomposition.
\end{remark}

\subsection{Local Hermitian bases and Bloch vectors}

For $S\in\{A,B\}$, fix $\Lambda^S_0=\Id{d_S}$ and let $\{\Lambda^S_k\}_{k=1}^{d_S^2-1}$ be the generalized Gell-Mann matrices of $\mathfrak{su}(d_S)$: traceless Hermitian matrices which, together with the identity, form an orthogonal basis of the real vector space of $d_S\times d_S$ Hermitian matrices,
\[
  \Tr\bigl(\Lambda^S_k\Lambda^S_{k'}\bigr)=2\,\delta_{kk'}, \qquad k,k'\ge 1 .
\]
For $d_S=2$ the Gell-Mann matrices are the Pauli matrices
\[
  X=\begin{pmatrix}0&1\\1&0\end{pmatrix},\quad
  Y=\begin{pmatrix}0&-i\\i&0\end{pmatrix},\quad
  Z=\begin{pmatrix}1&0\\0&-1\end{pmatrix},
\]
and every qubit density matrix can be written uniquely as
\[
  \rho=\tfrac12\bigl(\Id{2}+\vect{x}\cdot\vect{\sigma}\bigr),
  \qquad \vect{x}\cdot\vect{\sigma}=x_1X+x_2Y+x_3Z,
  \qquad x_k=\Tr(\sigma_k\rho)\in\R ;
\]
the vector $\vect{x}\in\R^3$ is the Bloch vector of $\rho$ and $\vect{\sigma} =(X,Y,Z)$. The Pauli matrices square to $\Id{2}$ and pairwise anticommute, so
\[
  (\vect{x}\cdot\vect{\sigma})^2
  =\sum_{i,j}x_ix_j\,\sigma_i\sigma_j
  =\sum_i x_i^2\,\Id{2}
  =|\vect{x}|^2\,\Id{2},
\]
the cross terms cancelling in pairs. Hence the eigenvalues of $\vect{x}\cdot\vect{\sigma}$ are $\pm|\vect{x}|$, and those of $\rho$ are $\lambda_\pm=\tfrac12(1\pm|\vect{x}|)$. Since $\rho\succeq 0$ we must have $|\vect{x}|\le 1$, and conversely.

\bigskip
\section{Main results}
\label{sec:main}
\bigskip
This section summarizes the main results of this paper. They are all constructive. 

\begin{theorem}[Separability with pure states]
\label{thm:main_pure}
Let $\rho$ be a density matrix on $\C^2\otimes\C^n$ with $\OSR(\rho)=3$. Then $\rho$ is separable with $\spr(\rho)=\rk{\rho}$.
\end{theorem}

An equivalent result was proved in \cite[Thm. 2]{Kraus2000}, where the decomposition is obtained by recursively subtracting product states from $\rho$, so as to lower the ranks of $\rho$ and $\rho^{T_A}$. We give an alternative proof based on an explicit construction, built from the eigendecomposition of a suitably chosen unitary. Besides being constructive, the method offers geometric insight into separability, and it extends to decompositions into mixed states.

\begin{theorem}[Separability with mixed states]
\label{thm:main_mixed}
Let $\rho$ be a density matrix on $\C^2\otimes\C^n$ with $\OSR(\rho)=3$. Then $\rho$ is separable with $\sr(\rho)\leq n+1$.
\end{theorem}

There exist states on $M_2(\C)\otimes M_n(\C)$ whose minimal separable decompositions require exactly $n+1$ terms (see \cref{ex:n+1}), so the bound in \cref{thm:main_mixed} is tight.

\begin{corollary}[The $3\times n$ case]\label{cor:3n}
Let $\rho$ be a density matrix on $\C^3\otimes\C^n$ with $\OSR(\rho)=3$ and $\rho=\rho^{T_A}$. Then $\rho$ is separable, with $\sr(\rho)\leq 2(n+1)$.
\end{corollary}

This corollary follows from a decomposition of \cite{Cariello_2021}. Cariello proved that on $\C^3\otimes\C^n$ states with $\OSR(\rho)=3$ and $\rho=\rho^{T_A}$, are separable and can be written as a sum of two terms containing embeddings of states on $\C^2\otimes\C^n$, each with $\OSR = 3$. Applying \cref{thm:main_mixed} to each summand directly yields the bound.

For two-qubit systems, \cref{thm:main_mixed} yields the following.

\begin{corollary}[Two qubits]\label{cor:two_qubit}
Let $\rho$ be a density matrix on $\C^2\otimes\C^2$ with $\OSR(\rho)=3$. Then $\rho$ is separable with $\sr(\rho)=3$.
\end{corollary}

Indeed, \cref{thm:main_mixed} gives $\sr(\rho)\leq 3$, and since $3=\OSR(\rho)\leq\sr(\rho)$ the equality follows. Combined with \cite{cariello2014,De_las_Cuevas_2019} for the case $\OSR(\rho)=\sr(\rho) = 2$ and \cite{SanperaTarrachVidal1998,Wootters1998} for $\OSR(\rho)=\sr(\rho) = 4$, we recover the result $\sr(\rho)=\OSR(\rho)$ of \cite{Jevtic2014}.

The proofs occupy the remainder of the paper. In \cref{sec:setup} SLOCC congruences, which preserve both separability and the operator Schmidt rank, reduce any $\OSR(\rho)=3$ state on $\C^2\otimes\C^n$ to the normal form
\[
  \rho=\tfrac{1}{2n}\bigl(\Id 2\otimes\Id n + X\otimes C_1+Z\otimes C_3\bigr),
  \qquad \text{with } C_1,C_3\ \text{Hermitian}.
\]
The two nontrivial Schmidt factors can be encoded in the single matrix $T=C_1+iC_3$. \cref{sec:jnr} studies its numerical range $W(T)$, establishing the geometric properties used throughout. 
\cref{thm:main_pure,thm:main_mixed} then correspond to two matrix dilations of $T$ detailed in \cref{sec:decomp}. Finally, the sharper results for the  two-qubit case (\cref{cor:two_qubit}) are discussed in \cref{sec:n2}, where a barycentric/Poncelet geometric construction (\cref{prop:threeterm}) yields the optimal three-term decomposition. To keep the construction self-contained, \cref{app:schmidt,app:cariello_construction,app:isometry} gather the technical details required for a full implementation that are omitted for readability.

\bigskip
\section{Reduction to a normal form}\label{sec:setup}
\bigskip

\subsection{Setup and tools }

From now on,
\[
  \Hs_A=\C^2 \ \text{(qubit)},\qquad
  \Hs_B=\C^n,\ n\ge 2 \ \text{(qudit)},\qquad
  \Hs_{AB}=\Hs_A\otimes\Hs_B\cong\C^{2n}.
\]
Let $\rho$ be a density matrix on $\Hs_{AB}$ with $\OSR(\rho)=3$. We seek a separable decomposition of the form $\rho=\sum_k p_k\,\alpha_k\otimes\beta_k$.

\begin{lemma}
\label{lem:osr3_mixed}
If $\OSR(\rho)\geq 2$, then the marginals $\rho_A=\Tr_B\rho$ and $\rho_B=\Tr_A\rho$ are both mixed.
\end{lemma}
\begin{proof}
If $\rho_A$ is pure, $\rho_A = \ketbra{a}{a}$. Then $\Tr[((I - \ketbra{a}{a})\otimes I )\rho]  = 0$ and $\rho$ is supported on $\text{span} \{\ket{a}\}\otimes \Hs_B$. Hence $\rho = \ketbra{a}{a} \otimes \rho_B$ and $\OSR(\rho) = 1$. The same argument shows $\rho_B$ cannot be pure. 
\end{proof}

\begin{lemma}[Invariance under SLOCC filtering \cite{Verstraete2001}]
\label{lem:filtering}
Let $V=V_1\otimes V_2$ with $V_1\in\mathrm{GL}_2(\C)$ and $V_2\in\mathrm{GL}_n(\C)$, and define the filtering map on density matrices $\rho\in M_{2n}(\C)$ by
\[
  \Phi_V(\rho)=\frac{V^\dagger\rho\,V}{\Tr\!\left(V^\dagger\rho\,V\right)}.
\]
Then $\Phi_V$ is well defined and maps density matrices to density matrices. It preserves $\OSR$, rank, separability, the separable rank, the pure separable rank, and the rank of each factor in a separable decomposition. Moreover $\Phi_V$ is invertible with $\Phi_V^{-1}=\Phi_{V^{-1}}$, so a separable decomposition of $\Phi_V(\rho)$ also yields one for $\rho$.
\end{lemma}

\begin{proof}
Since $V$ is invertible, $VV^\dagger$ is positive definite, so $\Tr(V^\dagger\rho\,V)=\Tr(\rho\,VV^\dagger)>0$ whenever $\rho\neq 0$; the denominator is therefore nonzero. For any $\ket{\psi}\in\C^{2n}$,
\[
  \bra{\psi}V^\dagger\rho\,V\ket{\psi}=(V\ket{\psi})^\dagger\rho\,(V\ket{\psi})\ge 0,
\]
so $V^\dagger\rho\,V$ is positive semi-definite. It is Hermitian, and normalizing by its (positive) trace gives a unit-trace PSD matrix, i.e.\ a density matrix.

If $\rho=\sum_i p_i\,\alpha_i\otimes\beta_i$, then, using $V^\dagger=V_1^\dagger\otimes V_2^\dagger$,
\[
  V^\dagger\rho\,V=\sum_i p_i\,(V_1^\dagger\alpha_i V_1)\otimes(V_2^\dagger\beta_i V_2).
\]
Thus the length of any decomposition is preserved, as is the $\OSR$. If the $\alpha_i,\beta_i$ are density matrices, each $V_1^\dagger\alpha_i V_1$ and $V_2^\dagger\beta_i V_2$ is PSD by the argument above. Renormalizing each local factor to unit trace and absorbing the positive scalars into the weights shows that $\Phi_V(\rho)$ is separable. The maps $\alpha\mapsto V_1^\dagger\alpha V_1$ and $\beta\mapsto V_2^\dagger\beta V_2$ are congruences, since $V_1,V_2$ are invertible, and congruence preserves rank; hence each term keeps the rank of both local factors. In particular, a rank-one (pure) product term is sent to a rank-one product term, so pure decompositions map to pure decompositions.

Because the two local congruences are bijections on the PSD cone (up to positive scalars), a decomposition with $r$ terms maps to one with exactly $r$ terms. Hence $\Phi_V$ preserves both the separable rank and the pure separable rank. Finally, we observe that $(V^{-1})^\dagger(V^\dagger\rho\,V)V^{-1}=\rho$ up to normalization, so $\Phi_V^{-1}=\Phi_{V^{-1}}$. 
\end{proof}

\subsection{Reduction to a PT-invariant state}

Our analysis relies on a result of Cariello \cite{cariello2014}, according to which every $\rho$ on $\Hs$ with operator Schmidt rank $\OSR(\rho)=3$ is separable. The argument of \cite{cariello2014} combines two ingredients: an invertible local filtering first maps $\rho$ to a state invariant under partial transposition, and such a state is in turn separable.

\begin{theorem}[{\cite[Theorem~3.2]{cariello2014}}]\label{thm:cariello}
Let $\rho$ be a density matrix on $\Hs$ with $\OSR(\rho)=3$. Then there exists an invertible $V=V_1\otimes V_2$ such that
\[
  \Phi_V(\rho)=\Phi_V(\rho)^{T_A},
\]
where $\Phi_V$ is the local filtering associated with $V$ (\cref{lem:filtering}) and $T_A$ denotes partial transposition on subsystem $A$.
\end{theorem}

The construction of $V$ is reproduced in \cref{app:cariello_construction}. Invariance under $T_A$ then yields separability through the following result.

\begin{theorem}[\cite{Kraus2000}]\label{thm:kraus}
Let $\rho$ be a state on $\C^2\otimes\C^n$. If $\rho=\rho^{T_A}$, then $\rho$ is separable.
\end{theorem}

 We reprove this theorem thanks to the decomposition obtained in \cref{thm:2n,thm:np1}.
Combining \cref{thm:cariello,thm:kraus} with the SLOCC-invariance of the separable rank (\cref{lem:filtering}), it suffices to construct separable decompositions for states satisfying $\rho=\rho^{T_A}$.

\subsection{Reduction to a maximally mixed B-marginal state}

\begin{lemma}[Reduction to full local rank \cite{li2018}]
\label{lem:full_local_rank_reduction}
Let $\rho$ be a density matrix on $\Hs_A\otimes\Hs_B= \mathbb{C}^2\otimes\mathbb{C}^n$ with $B$-marginal $\rho_B=\Tr_A\rho$, and suppose $\rk\rho_B=m<n$. Set $\Hs_C=\C^m$. Then there exist a unitary $U\in M_n(\C)$ and a density matrix $\rho'$ on $\Hs_A\otimes\Hs_C$ with $\Tr_A\rho'$ of full rank, such that
\[
  \rho'=\begin{pmatrix} \rho'_{1,1} & \rho'_{1,2} \\ \rho'_{2,1} & \rho'_{2,2} \end{pmatrix}
   \text{ and }
  (\Id{2}\otimes U^\dagger)\,\rho\,(\Id{2}\otimes U) 
  =\begin{pmatrix} \rho'_{1,1} & 0 & \rho'_{1,2} & 0 \\ 0 & 0 & 0 & 0 \\ \rho'_{2,1} & 0 & \rho'_{2,2} & 0 \\ 0 & 0 & 0 & 0 \end{pmatrix}.
\]
The last block structure is $\C^2\otimes(\C^m\oplus\C^{n-m})$ and $\rho'_{ij}$ are $m\times m$ matrices.
\end{lemma}

\begin{proof}
Diagonalize the $B$-marginal as $\rho_B=U\bigl(\Lambda\oplus 0_{n-m}\bigr)U^\dagger$, where $\Lambda\in M_m(\C)$ is positive definite and $U\in M_n(\C)$ is unitary. Set
\[
  \tilde\rho=(\Id{2}\otimes U^\dagger)\,\rho\,(\Id{2}\otimes U).
\]
Its $B$-marginal is $\Tr_A\tilde\rho=U^\dagger\rho_B U=\Lambda\oplus 0_{n-m}$. Let $Q=0_m\oplus\Id{n-m}$ be the projector onto the kernel of $\Tr_A\tilde\rho$. Then
\[
  \Tr\!\bigl[(\Id{2}\otimes Q)\,\tilde\rho\bigr]
  =\Tr_B\bigl[Q\,\Tr_A\tilde\rho\bigr]
  =\Tr_B\bigl[(0_m\oplus\Id{n-m})(\Lambda\oplus 0_{n-m})\bigr]
  =0.
\]
Since $\tilde\rho\succeq 0$ and $\Id{2}\otimes Q\succeq 0$, the vanishing of this trace forces $(\Id{2}\otimes Q)\,\tilde\rho=0$. Writing $\Pi=\Id{m}\oplus 0_{n-m}=\Id{n}-Q$, we obtain
\[
  \tilde\rho=(\Id{2}\otimes \Id{n})\,\tilde\rho\,(\Id{2}\otimes \Id{n})=(\Id{2}\otimes\Pi)\,\tilde\rho\,(\Id{2}\otimes\Pi)
  =\begin{pmatrix} \rho'_{1,1} & 0 & \rho'_{1,2} & 0 \\ 0 & 0 & 0 & 0 \\ \rho'_{2,1} & 0 & \rho'_{2,2} & 0 \\ 0 & 0 & 0 & 0 \end{pmatrix}
\]
in the decomposition $\C^2\otimes(\C^m\oplus\C^{n-m})$ with $\rho'_{ij}$ of dimension $m\times m$. Setting
$\rho'=\begin{psmallmatrix}\rho'_{1,1} & \rho'_{1,2}\\ \rho'_{2,1} & \rho'_{2,2}\end{psmallmatrix}$
gives $\rho'\in M_2(\C)\otimes M_m(\C)$ with $\rho'\succeq 0$ and $\Tr\rho'=1$. Since $\Tr_A\rho'=\Lambda\succ 0$, its $B$-marginal has full rank $m$.
\end{proof}

This construction uses the unitary filtering $\Phi_{\Id{2}\otimes U}$, a special case of \cref{lem:filtering}. Therefore, it preserves separability, the separable rank, and the rank of each factor. Moreover, a separable decomposition of $\rho'$ yields one of $\tilde\rho$ by padding each $B$-factor with zeros, which leaves the rank of every term unchanged. 

Suppose now that $\rho_B$ has full rank; then $\rho_B$ is positive definite, hence invertible, and admits a unique positive-definite Hermitian square root $\tau_B$ with $\tau_B^2=\rho_B$. Consider the filtering $\Phi_V$ with $V=\Id{2}\otimes\tau_B^{-1}$. Because $V$ acts as the identity on subsystem $A$ while $T_A$ transposes only that subsystem, the two operations commute; hence $\Phi_V$ preserves the partial-transpose invariance $\rho=\rho^{T_A}$. Its $B$-marginal is
\[
  \Tr_A\!\bigl(V^\dagger\rho\,V\bigr)
  =\tau_B^{-1}\rho_B\,\tau_B^{-1}
  =\Id{n},
\]
so after normalization $\Tr_A\Phi_V(\rho)=\Id{n}/n$. We may thus restrict to the case $\rho_B=\Id{n}/n$. 

\subsection{Normal form and the matrix T}

The discussion of the previous paragraphs shows we can assume, without loss of generality, throughout
the rest of the paper that $\rho$ is a density matrix satisfying:
\begin{description}
  \item[(H1) PT-invariance] $\rho^{T_A}=\rho$, where $T_A$ is the partial transpose on system $A$ in the computational basis. Since $X^T=X$, $Z^T=Z$, and $Y^T=-Y$, this is equivalent to the vanishing of the $Y^A$ correlations, i.e.\ $\Tr\bigl(\rho\,(Y\otimes\Lambda^B_j)\bigr)=0$ for all $j\ge 0$.
  \item[(H2) Maximally mixed $B$-marginal] $\rho_B=\Tr_A\rho=\Id{n}/n$.
  \item[(H3) Mixed $A$-marginal] $\rho_A=\Tr_B\rho=\tfrac12\bigl(\Id{2}+a_1X+a_3Z\bigr)$ with $\vect{a}=(a_1,a_3)\in\R^2$. The absence of a $Y$-term follows from (H1), and since by \cref{lem:osr3_mixed} $\rho_A$ is mixed, we have $a_1^2+a_3^2<1$.
\end{description}

\begin{lemma}[Normal form]
\label{lem:normalform}
Under hypotheses (H1)--(H3) there exist unique Hermitian matrices $C_1,C_3\in M_n(\C)$ such that
\begin{equation}
\label{eq:rhoC}
  \rho=\frac{1}{2n}\bigl(\Id{2}\otimes\Id{n}+X\otimes C_1+Z\otimes C_3\bigr).
\end{equation}
\end{lemma}

\begin{proof}
Expand $\rho$ in the orthogonal basis $\{\sigma_i\otimes\Lambda^B_j\}_{i,j\ge 0}$ of Hermitian matrices on $\Hs$, where $(\sigma_0,\sigma_1,\sigma_2,\sigma_3)=(\Id{2},X,Y,Z)$. Hypothesis (H1) removes every term with $\sigma_i=Y$, and (H2) removes the terms $\Id{2}\otimes\Lambda^B_j$ for $j\ge1$, since $\Tr\bigl(\rho\,(\Id{2}\otimes\Lambda^B_j)\bigr)=\Tr(\rho_B\Lambda^B_j)=0$. The surviving terms are $\Id{2}\otimes\Id{n}$, together with $X\otimes\Lambda^B_j$ and $Z\otimes\Lambda^B_j$ for $j\ge0$; collecting the $B$-factors, we obtain
\[
  \rho=\frac{1}{2n}\Bigl(
    \Id{2}\otimes\Id{n}
    +X\otimes C_1
    +Z\otimes C_3
  \Bigr).
\]
The matrices $C_1, C_3$ are given by their (unique) expansion on the orthonormal Gell-Mann basis. 
They are generally not traceless, with $\tfrac1n\Tr C_k=a_k$.
\end{proof}

Operationally, $\tfrac1n\,C_k=\Tr_A\bigl(\rho\,(\sigma_k\otimes\Id{n})\bigr)$ is the $B$-operator obtained by measuring $\sigma_k$ on $A$. Conversely, any Hermitian $C_1, C_3$ in \eqref{eq:rhoC} produces a Hermitian operator satisfying (H1) and (H2), with $\vect{a}=\bigl(\tfrac1n\Tr C_1,\tfrac1n\Tr C_3\bigr)$; only the positivity $\rho\succeq0$ is a nontrivial constraint, and it is characterized next.

\subsection{Positivity is equivalent to  contraction}
\begin{definition}\label{def:T}
Given $\rho$ as in \eqref{eq:rhoC}, we associate the complex matrix
\[
  T=C_1+i\,C_3\in M_n(\C).
\]
\end{definition}

\begin{proposition}
\label{prop:contraction}
Let $\rho$ be as in \eqref{eq:rhoC} with $C_1,C_3$ Hermitian. Then
\[
  \rho\succeq 0 \iff \norm{T}\le 1,
\]
i.e. $T$ is a contraction. Furthermore, $\rk(\rho) = n + \rk(I-TT^\dagger)$.
\end{proposition}
\begin{proof}
We work with $2n\rho=\Id{2}\otimes\Id{n}+X\otimes C_1+Z\otimes C_3$. Let $U=\tfrac1{\sqrt2}(\Id{2}+iX)\in M_2(\C)$, a unitary with $UXU^\dagger=X$ and $UZU^\dagger=Y$. Hence, with $V=U\otimes\Id{n}$,
\[
  V(2n\rho)V^\dagger=\Id{2}\otimes\Id{n}+X\otimes C_1+Y\otimes C_3 .
\]
In the block decomposition induced by the computational basis $\{\ket0,\ket1\}$ of the qubit,
\[
  X\otimes C_1=\begin{pmatrix}0&C_1\\C_1&0\end{pmatrix},
  \qquad
  Y\otimes C_3=\begin{pmatrix}0&-iC_3\\iC_3&0\end{pmatrix}.
\]
Using the Schur complement, we obtain
\begin{equation}
\label{eq:blockT}
  V(2n\rho)V^\dagger
  =\begin{pmatrix}\Id{n} & T^\dagger\\ T & \Id{n}\end{pmatrix}
    = \begin{pmatrix} \Id{n} & 0 \\ T & \Id{n} \end{pmatrix}
    \begin{pmatrix} \Id{n} & 0 \\ 0 & \Id{n} - T T^\dagger \end{pmatrix}
    \begin{pmatrix} \Id{n} & T^\dagger \\ 0 & \Id{n} \end{pmatrix}.
\end{equation}
Thus
\[
  \begin{pmatrix}\Id{n}&T^\dagger\\T&\Id{n}\end{pmatrix}\succeq 0
  \iff \Id{n}-TT^\dagger\succeq 0
  \iff \norm{T}\le 1 .
\]
The last equivalence holds because $\Id{n}-TT^\dagger\succeq0$ means every singular value of $T$ is at most $1$. From the block form \eqref{eq:blockT}, $\rk(\rho)=n+\rk(\Id n-TT^\dagger)$.
\end{proof}

\begin{remark}
The Hermitian operator $X\otimes C_1+Z\otimes C_3$ is unitarily equivalent to 
\[
  \begin{pmatrix} 0 & T \\ T^\dagger & 0 \end{pmatrix},
\]
whose spectrum is $\{\pm s_k(T)\}_{k=1}^n$, where $s_k(T)$ are the singular values of $T$. Consequently,
\[
  \operatorname{spec}(\rho)=\Bigl\{\tfrac{1\pm s_k(T)}{2n}\Bigr\}_{k=1}^n,
\]
and in particular $\lmin(\rho)=\frac{1-s_{\max}(T)}{2n}$. Hence $\rho$ has a zero eigenvalue if and only if $s_{\max}(T)=\norm{T}=1$.
\end{remark}

\subsection{Associated definitions}

\begin{definition}
    A contraction $T$ is completely non-unitary (c.n.u.) if the only vector $x$ with $\|T^k x\| = \|x\| = \|T^{\dagger k} x\|$ for every  $k \geq 1$ is $x = 0$.
\end{definition}

\begin{lemma}[Canonical decomposition]\label{lem:decomposition_cnu}
Let $T$ be a contraction on $\C^n$. Then there is a unique decomposition $\Hs=\Hs_{0}\oplus \Hs_{1}$ into reducing subspaces such that $U=T|_{\Hs_{0}}$ is unitary and $A=T|_{\Hs_{1}}$ is completely non-unitary. Hence $T=U\oplus A$. Moreover $\Hs_{0}$ is the largest reducing subspace on which $T$ is unitary.
\end{lemma}

We refer to \cite[Thm 3.2]{sznagy2010harmonic} for the proof. For a contraction $T$ on $\mathbb{C}^{n}$,
\[
  \Hs _{0}=\bigoplus_{|\lambda|=1}\ker(T-\lambda I),
\]
the sum being over the unimodular eigenvalues of $T$.

\begin{definition}
    For a contraction $T$ on $\C^n$, the \emph{defect operator} is
\[
  D_T = (I-T^\dagger T)^{1/2}\ \succeq 0,
\]
and the \emph{defect index} is $d = \rk D_T = \rk D_{T^\dagger}$.

\end{definition}

\bigskip
\section{The joint numerical range}
\label{sec:jnr}
\bigskip

\subsection{Definition and Properties}

\begin{definition}
\label{def:jnr}
The \emph{joint numerical range} of the pair $(C_1,C_3)$ is
\[
  \W=\Bigl\{
    \bigl(\braket{\psi|C_1|\psi},\,\braket{\psi|C_3|\psi}\bigr)\in\R^2
    \;:\;
    \ket\psi\in\C^n,\;\braket{\psi|\psi}=1
  \Bigr\}.
\]
Identifying $\R^2\cong\C$, this is the ordinary numerical range of $T = C_1 + i C_3$:
$\W=W(T)=\{\braket{\psi|T|\psi}, \braket{\psi | \psi}=1\}$.
\end{definition}

\begin{remark}
    For $n=2$, thanks to the normal form, this is the steering ellipse \cite{Jevtic2014} of the filtered state in the plane given by $X$ and $Z$.
\end{remark}

An extensive review of numerical ranges  is done in \cite{Gau_Wu_2021}. We present here the properties relevant to the problem.

\begin{figure}
  \centering
  \begin{subfigure}[b]{0.64\textwidth}
    \centering
    \includegraphics[width=\linewidth, trim={0 0 400 0},clip ]{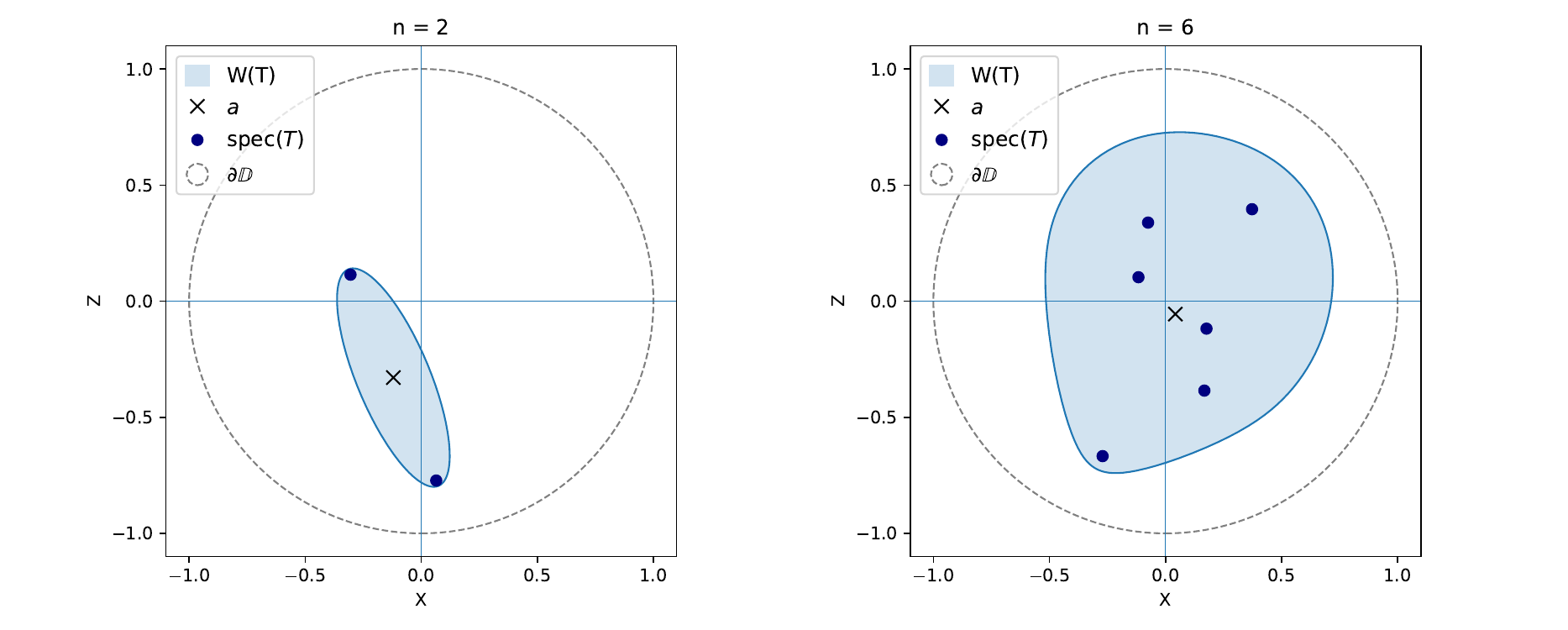}
    \caption{W(T)}
    \label{fig:T}
  \end{subfigure}
  \hfill
  \begin{subfigure}[b]{0.32\textwidth}
    \centering
    \includegraphics[width=\linewidth, trim={850 0 0 0},clip ]{img/numerical_range_523.pdf}
    \caption{W(U)}
    \label{fig:U}
  \end{subfigure}
  \caption{Numerical ranges of two contractions $T$ with $n=2,6$ and a unitary $U$ with $n=5$. $W(T)$ is convex and $\vect{a}$ is inside $W(T)$.}
  \label{fig:jnr}
\end{figure}

\begin{lemma}[Mixed-state extension]
\label{lem:W_mixed}
Let $\widetilde\W=\bigl\{\bigl(\Tr(\tau C_1),\Tr(\tau C_3)\bigr): \tau\succeq 0,\;\Tr\tau=1\bigr\}$ be the mixed-state joint numerical range. Then $\widetilde\W=\conv(\W)$.
\end{lemma}

\begin{proof}
($\subseteq$) Let $\tau=\sum_i p_i\ketbra{\phi_i}{\phi_i}$ be a spectral
decomposition, $p_i\ge0$, $\sum_ip_i=1$. By linearity of the trace,
\[
  \bigl(\Tr(\tau C_1),\Tr(\tau C_3)\bigr)
  =\sum_i p_i\bigl(\braket{\phi_i|C_1|\phi_i},\braket{\phi_i|C_3|\phi_i}\bigr),
\]
a convex combination of points of $\W$.

($\supseteq$) Let $w=\sum_i p_iw_i\in\conv(\W)$ with $w_i\in\W$ realized by unit
vectors $\ket{\phi_i}$. Then $\tau=\sum_i p_i\ketbra{\phi_i}{\phi_i}$ is a density
matrix and, again by linearity, $\bigl(\Tr(\tau C_1),\Tr(\tau C_3)\bigr)=w$.
\end{proof}

\begin{lemma}[Properties of $\W$]
\label{lem:Wprops}
Assume $\rho$ of the form \eqref{eq:rhoC} is positive semi-definite. Then:
\begin{enumerate}
  \item $\W$ is compact and convex;
  \item $\W\subseteq\Disk$;
  \item $\vect a=(a_1,a_3)\in\W$, where $\vect a$ is the Bloch vector of $\rho_A$.
\end{enumerate}
\end{lemma}

These properties are illustrated in \cref{fig:T}. For $n=2$,  $\W$ is an ellipse, whose foci are the eigenvalues of $T$ by the Elliptical range theorem, \cite{Li1996}.

\begin{proof}
\emph{(1) Compactness.} The unit sphere $S=\{\ket\phi\in\C^n:\braket{\phi|\phi}=1\}$ is compact (in finite dimension), and $f(\ket\phi)=(\braket{\phi|C_1|\phi},\braket{\phi|C_3|\phi})$ is continuous, so $\W=f(S)$ is compact.

\emph{Convexity.} Follows from the Toeplitz--Hausdorff theorem \cite{Toeplitz-hausdorf}.

\emph{(2) Containment in $\Disk$.} Follows directly from $\norm{T}\leq 1$ by \cref{prop:contraction}.

\emph{(3) $\vect a\in\W$.} By the normal form,  $\vect a=\bigl(\Tr(\tfrac{\Id n}{n}C_1),\Tr(\tfrac{\Id n}{n}C_3)\bigr) \in\widetilde\W$. By \cref{lem:W_mixed} and by convexity of $\W$, $\widetilde\W=\W$, so $\vect a\in\W$.
\end{proof}

\subsection{General facts on numerical ranges}

The following results concern the ordinary numerical range $W(\cdot)$ of an arbitrary operator and do not rely on properties of $\rho$.

\begin{lemma}[Unitary invariance and direct sums]\label{lem:W_general}
\begin{enumerate}
    \item[]
    \item For every unitary $U\in M_n(\C)$, $\;W(U^\dagger M U)=W(M)$.
    \item For $M_1\in M_{n_1}(\C)$, $M_2\in M_{n_2}(\C)$,
        $\;W(M_1\oplus M_2)=\conv\bigl(W(M_1)\cup W(M_2)\bigr)$.
\end{enumerate}
\end{lemma}

\begin{proof}
\emph{(1)} The map $\ket\psi\mapsto U\ket\psi$ is a bijection of the unit sphere, and $\braket{\psi|U^\dagger M U|\psi}=\braket{U\psi|M|U\psi}$, so the two sets of values coincide.

\emph{(2)} Write a unit vector as $\ket\psi=(\ket{\psi_1},\ket{\psi_2})$ with
$\norm{\psi_1}^2+\norm{\psi_2}^2=1$. Then
\[
  \braket{\psi|(M_1\oplus M_2)|\psi}
  =\norm{\psi_1}^2\braket{\hat\psi_1|M_1|\hat\psi_1}
  +\norm{\psi_2}^2\braket{\hat\psi_2|M_2|\hat\psi_2},
  \qquad \hat\psi_i=\ket{\psi_i}/\norm{\psi_i},
\]
a convex combination of a point of $W(M_1)$ and a point of $W(M_2)$ (with only one term when one block vanishes), hence $W(M_1\oplus M_2)\subseteq\conv(W(M_1)\cup W(M_2))$. Conversely, vectors supported in a single block give $W(M_1),W(M_2)\subseteq W(M_1\oplus M_2)$, and the latter is convex by the Toeplitz-Hausdorff theorem, so it contains
$\conv(W(M_1)\cup W(M_2))$.
\end{proof}

\begin{theorem}
\label{thm:unitary_polygon}
Let $\lambda_1,\dots,\lambda_n$ be the eigenvalues of a unitary $U$ (listed with multiplicity). Then
\[
  W(U)=\conv\{\lambda_1,\dots,\lambda_n\},
\]
a closed convex polygon whose vertices lie on the unit circle.
\end{theorem}

\begin{proof}
A unitary operator is normal, so by the finite-dimensional spectral theorem there is an orthonormal basis $e_1,\dots,e_n$ of $\C^n$ consisting of eigenvectors, $Ue_k=\lambda_k e_k$. Unitarity forces $|\lambda_k|=1$, so every eigenvalue lies on the unit circle.

Let $\ket{\phi}$ be any unit vector and expand $\ket{\phi}=\sum_k c_k \ket{e_k}$, so that $\|\phi\|^2=\sum_k|c_k|^2=1$. Using $\langle e_k | e_j\rangle=\delta_{kj}$,
\[
  \braket{\phi| U| \phi}
  = \sum_k \sum_j  \bar{c_k} c_j\lambda_k \braket{e_k | e_j}
  =\sum_k |c_k|^2\lambda_k,
\]
a convex combination of the $\lambda_k$ with weights $|c_k|^2$; hence $\braket{\phi| U| \phi}\in\conv\{\lambda_1,\dots,\lambda_n\}$.

Conversely, let $\mu=\sum_k t_k\lambda_k$ with $t_k\ge0$, $\sum_k t_k=1$. Set $\ket{\psi}=\sum_k\sqrt{t_k}\,\ket{e_k}$, then $\|\psi\|^2=1$ and, by the previous computation with $c_k=\sqrt{t_k}$, $\braket{\psi| U| \psi}=\mu$, so $\mu\in W(U)$.
\end{proof}

\begin{remark} This follows from a more general fact: for any normal matrix $N$ over $\C^n$, we have $W(N) = \conv( \operatorname{spec} (N))$. Furthermore, unitarity implies that the eigenvalues lie on the unit circle. See \cref{fig:U} for an example.
\end{remark}

\subsection{Dilation}
\begin{definition}[Dilation]
    Let $A$ be an operator on a Hilbert space $H$. A \emph{ dilation} (resp. \emph{power dilation}) of $A$ is an operator $B$ on a larger Hilbert space $K\supseteq H$ if there exists an isometry $J : H \mapsto  K$ such that 
\[
  A = J^\dagger\,B J\qquad (\text{resp. } A^k  = J^\dagger\,B^k J\text{ for all }  k ),
\]
where $\Pi_H = JJ^\dagger$ is the orthogonal projection of $K$ onto $H$ and $J^\dagger J = \Id{H}$.
\end{definition}

When $B$ is unitary, the dilation is said to be a unitary dilation. This is possible only if $A$ is a contraction.

\begin{lemma}
\label{lem:dilation}
Let $B$ be a dilation of $A$. Then $W(A)\subseteq W(B)$.
\end{lemma}
\begin{proof}
The dilation condition reads $A=J^\dagger BJ$.

Let $w\in W(A)$, realized by a unit vector $\ket{\phi}\in H$, $\braket{\phi|A|\phi}=w$. Its image $J\ket{\phi}\in K$ is again a unit vector, since $\|J\ket{\phi}\|^2=\braket{\phi|J^\dagger J|\phi}=\braket{\phi|\phi}=1$. Then
\[
  \braket{J\phi|B|J\phi}
  =\braket{\phi|J^\dagger BJ|\phi}
  =\braket{\phi|A|\phi}
  =w .
\]
Thus $w\in W(B)$, so $ W(A)\subseteq W(B)$.
\end{proof}

\bigskip
\section{Separable decompositions}\label{sec:decomp}
\bigskip

We seek decompositions $\rho=\sum_{k=1}^m p_k \,\alpha_k \otimes\beta_k $
with $p_k >0$, $\sum_k  p_k =1$, and $\alpha_k \in M_2(\C)$,
$\beta_k \in M_n(\C)$ density matrices.

\subsection{Moment conditions and geometric criterion}

\begin{lemma}[$xz$-reduction]\label{lem:xz}
Assume (H1). If $\rho$ admits a separable decomposition with $m$ terms, then it admits one with $m$ terms in which every $A$-factor has its Bloch vector in the $xz$-plane:
\[
  \alpha_k =\tfrac12(\Id2+x_k  X+z_k  Z).
\]
\end{lemma}
\begin{proof}
Let $\rho=\sum_k  p_k \alpha_k \otimes\beta_k $. In the computational basis, transposition acts on Bloch vectors by $(x,y,z) \mapsto (x,-y,z)$. By (H1),
\[
  \rho=\tfrac12(\rho+\rho^{T_A})
  =\sum_k  p_k \,\tilde\alpha_k \otimes\beta_k ,
  \qquad
  \tilde\alpha_k =\tfrac12(\alpha_k +\alpha_k ^{T}).
\]
Each $\tilde\alpha_k $ is a density matrix with its Bloch vector in the $xz$ plane.
\end{proof}

\begin{proposition}[Geometric criterion]\label{prop:geom}
Assume (H1)-(H3). If $\rho=\sum_{k =1}^m p_k \alpha_k \otimes\beta_k $ is separable with $A$-Bloch vectors $a_k  = (x_k ,y_k ,z_k )$, then
\[
  \W \subseteq\conv(v_1,\dots,v_m) \qquad \text{ where } v_k  = (x_k ,z_k ) \in \Disk.
\]
Consequently $\sr(\rho) $ is at least the minimum number of vertices in $\Disk$ necessary to enclose $\W (T)$.
\end{proposition}
\begin{proof}
By \cref{lem:xz} we may assume the decomposition is in $xz$-form, so $y_k=0$. Matching the $\Id2$-, $X$- and $Z$-components of $\rho=\tfrac12\sum_k  p_k (\Id2+x_k  X+z_k  Z)\otimes\beta_k $ with \eqref{eq:rhoC} gives the \emph{moment conditions}
\begin{equation}\label{eq:moments}
  \sum_k  p_k  \beta_k  =\frac{\Id n}{n},
  \qquad
  \sum_k  x_k  p_k  \beta_k =\frac{C_1}{n},
  \qquad
  \sum_k  z_k  p_k  \beta_k =\frac{C_3}{n}.
\end{equation}
For an arbitrary unit vector $\ket\psi\in\C^n$, the positivity of $\beta_k $ implies that $p_k  \braket{\psi |\beta_k  |\psi}$ is real and non negative for all $k $. The first condition gives  $\sum_k  n p_k  \braket{\psi |\beta_k  |\psi}=1$ and the other two give
\[
  \sum_k  n  p_k  \braket{\psi |\beta_k  |\psi} v_k  = \sum_k   n p_k  \braket{\psi |\beta_k  |\psi}  (x_k , z_k )
  =\bigl(\braket{\psi|C_1|\psi},\braket{\psi|C_3|\psi}\bigr)\in\W .
\]
Thus, every point of $\W$ is a convex combination of the $v_k=(x_k, z_k)$. 
\end{proof}

\begin{lemma}\label{lem:push}
    Assume (H1). If $\rho$ admits a separable decomposition with $m$ terms, then there exists a decomposition in $xz$-form with at most $m$ terms in which every $A$-factors is pure. 
\end{lemma}

\begin{proof}
    Consider the decomposition obtained by \cref{lem:xz}.
    Choose $k$ such that $v_k \in \oDisk$ and $j$ such that $v_k \neq v_j$ (always possible since $\OSR >1 $). 
    Let $v_k'\in \bD$ be the point where the ray from $v_j$ through $v_k$ intercepts $\bD$. Set $s= \frac{|v_k-v_j|}{|v_k'-v_j|} \in (0,1]$, such that $v_k = (1-s)v_j+ sv_k'$. 
    Then replace the $k$th and $j$th terms by 
    \begin{align}
        p_{k'} = sp_k,\quad \alpha_{k'} = \tfrac{1}{2}(I + x_k'X + z_k'Z),\quad \beta_{k'} =\beta_k, \\
    p_{j'} = p_j + (1-s)p_k, \quad \alpha_{j'}=\alpha_j , \quad \beta_{j'} = \tfrac{p_j\beta_j + (1 - s)p_k\beta_k}{p_j + (1-s)p_k}.
    \end{align}
    Both moment conditions are preserved, all terms stay PSD and we can iterate this procedure over all terms.
\end{proof}

Therefore, the separable rank is always attained with pure $A$-factors. 

\begin{lemma}\label{lem:dilation_decomposition}
Let $T$ be a contraction on a Hilbert space $\mathcal{H}$ (so $\|T\|\le 1$), and let $U$ be a unitary dilation of $T$, acting on a larger Hilbert space $\mathcal{K}\supseteq\mathcal{H}$, meaning that there exists an isometry $J$ such that
\[
  T = J^\dagger\,U J
\]
Suppose $U$ has $m$ distinct eigenvalues $e^{i\theta_1},\dots,e^{i\theta_m}$ with corresponding orthogonal projections $E_\mu$ onto the eigenspace of $e^{i \theta_\mu}$ for $1 \leq \mu \leq m $. Then the operators
\[
L_\mu = J^\dagger\,E_\mu J
\]
are hermitian, positive semi-definite, and satisfy $\sum_{\mu=1}^m L_\mu = \Id{\mathcal{H}}$,
$\sum_{\mu=1}^m \cos(\theta_\mu) L_\mu = C_1$, and $\sum_{\mu=1}^m \sin(\theta_\mu) L_\mu = C_3$.
\end{lemma}

\begin{proof}
By construction, $L_\mu$ is hermitian and positive semi definite since $E_\mu$ is an orthogonal projector. Since $\sum_\mu E_\mu=\Id{\mathcal{K}}$,
\[
\sum_\mu L_\mu
  = J^\dagger\Big(\sum_\mu E_\mu\Big)J
  = J^\dagger J
  = \Id{\mathcal{H}}.
\]

The dilation property gives $T = J^\dagger U J$.
Using $U=\sum_\mu e^{i \theta_\mu} E_\mu$,
\[
T = J^\dagger U J
  = J^\dagger\Big(\sum_\mu e^{i \theta_\mu} E_\mu\Big)J
  = \sum_\mu e^{i \theta_\mu}\, L_\mu
  = \sum_\mu \cos (\theta_\mu)\, L_\mu + i \sum_\mu \sin (\theta_\mu)\, L_\mu.
\]
Taking the Hermitian and anti-Hermitian part of this equation yields the searched form.
\end{proof}

\subsection{Pure product states decomposition}

\begin{lemma}\label{lem:defect_rank}
Let $T$ be a contraction on $\mathbb{C}^n$ with defect index $d$.
Then there is a unitary $U$ on $\mathbb{C}^{n+d}$ with
\[
T=J^\dagger\,U\,J,\qquad
J=\begin{pmatrix}I_n\\ 0_d\end{pmatrix}.
\]
\end{lemma}

\begin{proof} We reproduce here the construction of \cite{sznagy2010harmonic}, see chapter I.3-5 for more details.
Write $\DT=\operatorname{ran}D_T$, $\DTs=\operatorname{ran}D_{T^\dagger}$, both have dimension $d$. Let $V:\mathbb{C}^d\to\mathbb{C}^n$ and $V_*:\mathbb{C}^d\to\mathbb{C}^n$ be isometric embeddings with ranges $\DT$ and $\DTs$, such that
\[
V^\dagger V=I_d,\quad VV^\dagger=\Pi_{\DT},\qquad
V_*^\dagger V_*=I_d,\quad V_*V_*^\dagger=\Pi_{\DTs}.
\]

We have the intertwining relations
\begin{equation}\label{eq:intertwine}
D_T\,T^\dagger=T^\dagger D_{T^\dagger},\qquad
D_{T^\dagger}\,T=T\,D_T.
\end{equation}
Furthermore, since $\operatorname{ran}V_*=\DTs=\operatorname{ran}D_{T^\dagger}$, write $V_*=D_{T^\dagger}M$ for some $M:\mathbb{C}^d\to\mathbb{C}^n$. Then by \eqref{eq:intertwine}, $T^\dagger V_*=T^\dagger D_{T^\dagger}M =D_T\,T^\dagger M\in\DT$, so
\begin{equation}\label{eq:projector}
\Pi_{\DT}\,T^\dagger V_*=T^\dagger V_*.
\end{equation}

Define, on $\mathbb{C}^n\oplus\mathbb{C}^d$,
\[
U=\begin{pmatrix}
T & D_{T^\dagger}V_*\\[3pt]
V^\dagger D_T & -\,V^\dagger T^\dagger V_*
\end{pmatrix}.
\]

A direct computation of $U^\dagger U$ gives,
\[
U^\dagger U=
\begin{pmatrix}
T^\dagger T+D_T^2 & T^\dagger D_{T^\dagger}V_*-D_T\,T^\dagger V_*\\
\ast & V_*^\dagger D_{T^\dagger}^2 V_*+V_*^\dagger T\,\Pi_{\DT}\,T^\dagger V_*
\end{pmatrix}=I_{n+d}.
\]
Here the $(2,1)=*$ entry is the adjoint of the $(1,2)$. Since $T^\dagger T+D_T^2=T T^\dagger+D_{T^\dagger}^2=I_n$, and using \cref{eq:projector}, the diagonal block are, respectively, $\Id n$ and $\Id d$. The $(1,2)$ entry vanishes by \cref{eq:intertwine}. Finally from the $(1,1)$ entry of $U$ we see $J^\dagger U J=T$.
\end{proof}

No unitary dilation of $T$ acts on a space of dimension smaller than $n+d$. The resulting dilation is not unique: any unitary of the form $\begin{psmallmatrix}
    I_n & 0 \\ 0 & U_0\\
\end{psmallmatrix} U$ with $U_0 \in U(d)$ is a valid dilation of $T$ of size $n+d$.
When $T$ has full defect index, $V=V_* =\Id{n}$, so one recovers with this construction the Halmos dilation  \cite{Halmos1950}, the $2n\times2n$ matrix
\begin{equation*}
  U_H=\begin{pmatrix}T & D_{T^\dagger}\\ D_T & -T^\dagger\end{pmatrix}.
\end{equation*}

\begin{theorem}\label{thm:2n}
Assume (H1)-(H3) and $\rho\succeq0$. Then $\rho$ admits a separable decomposition
\begin{equation*}\label{eq:decomp2n}
  \rho=\sum_{k=1}^{\rk(\rho)}p_k\,\alpha_k\otimes\beta_k
\end{equation*}
with $\rk (\rho)$ terms, in which every $\alpha_k$ is a \emph{pure} qubit state and every $\beta_k$ a \emph{pure} state on $\C^n$. Thus, $\rho$ is separable and $\spr(\rho)= \rk (\rho)$.
\end{theorem}
\begin{proof}
By \cref{prop:contraction}, $T$ is a contraction with defect index $d=\rk(\rho)-n$.  Thanks to \cref{lem:defect_rank}, every contraction of defect index $d$ admits a unitary dilation $U$ of size $n + d = \rk(\rho)$. 
By \cref{lem:dilation,thm:unitary_polygon}, $\W$ is inside the polygon formed by the $\rk(\rho)$ eigenvalues of $U$.

Let $e^{i\theta_k}$, $k=1,\dots,\rk(\rho)$, be the eigenvalues of $U$ (all on $\bD$) with orthonormal eigenvectors
\[
  U\ket{u_k}=e^{i\theta_k}\ket{u_k},
  \qquad
  \ket{u_k}=\begin{pmatrix}\ket{\psi_k}\\ \ket{\chi_k}\end{pmatrix},
  \quad \ket{\psi_k}\in\C^n,\ket{\chi_k}\in\C^d .
\]
Let $J=\begin{psmallmatrix}I_n\\ 0_d\end{psmallmatrix}$ be the isometric embedding onto the top block, such that $T=J^\dagger UJ$. By \cref{lem:dilation_decomposition},
\begin{equation}\label{eq:halmosmoments}
  \sum_{k=1}^{\rk(\rho)}\ketbra{\psi_k}{\psi_k}=\Id n,
  \qquad
  \sum_{k=1}^{\rk(\rho)}\cos\theta_k\,\ketbra{\psi_k}{\psi_k}=C_1,
  \qquad
  \sum_{k=1}^{\rk(\rho)}\sin\theta_k\,\ketbra{\psi_k}{\psi_k}=C_3 .
\end{equation}
Define
\begin{equation}\label{eq:halmosterms}
  p_k=\frac{\|\psi_k\|^2}{n},
  \qquad
  \alpha_k=\tfrac12\bigl(\Id2+\cos\theta_k\,X+\sin\theta_k\,Z\bigr),
  \qquad
  \beta_k=\frac{\ketbra{\psi_k}{\psi_k}}{\|\psi_k\|^2}.
\end{equation}
 Dividing \eqref{eq:halmosmoments} by $n$ and using $p_k\beta_k=\ketbra{\psi_k}{\psi_k}/n$ reproduces exactly the moment conditions \eqref{eq:moments} with $v_k=(\cos\theta_k,\sin\theta_k)$. Hence $\rho=\sum_{k=1}^{\rk(\rho)}p_k\,\alpha_k\otimes\beta_k$.
The weights satisfy $p_k\ge0$ and $\sum_kp_k=\Tr(\Id n)/n=1$ by the first identity in \eqref{eq:halmosmoments}. 
Each $\alpha_k$ is a pure qubit state with Bloch vector $(\cos\theta_k,0,\sin\theta_k)$, each $\beta_k$ is a rank-one projector, hence a pure state on $\C^n$. Note that  $\rk(\rho)\leq \spr(\rho)$ (as a sum of rank one projectors) and this decomposition gives the equality.
\end{proof}

Examples of dilations are presented in \cref{fig:pure}.
\begin{figure}
    \centering
    \includegraphics[width=\linewidth, trim={10 20 10 30},clip ]{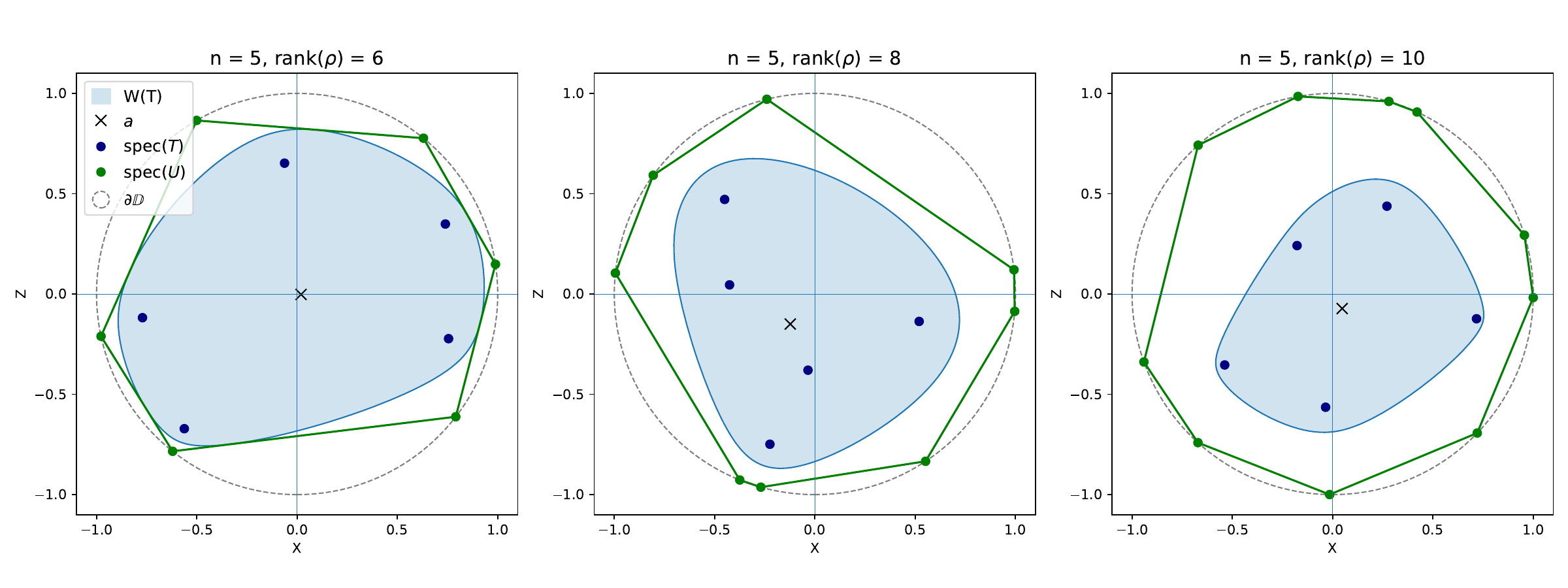}
    \caption{Numerical range for three bipartite system with $n=5$ and different values of $\rk(\rho)$. The dilation $U$ has dimension $\rk {\rho}$. The vertices of the enclosing polygon $W(U)$ correspond to eigenvalues of $U$. The numerical range $W(U)$ need not tightly enclose $W(T)$.}
    \label{fig:pure}
\end{figure}

\begin{remark}
Once the normal form and $C_1, C_3$ are determined, any {\it normal} dilation $N$ of $T=C_1+ iC_3$ with spectrum in $\Disk$ yields a valid decomposition. 
\end{remark}

\subsection{Mixed \texorpdfstring{$n+1$}{n+1}-states decompositions}\label{sec:mterm}

To construct a polygon with $n+1$ vertices that enclose $\W$, we rely on the following theorem of Wu.

Let $A$ be a completely non-unitary contraction on a finite-dimensional Hilbert space, so its minimal polynomial $m_A$ has all roots in the open unit disk $\oDisk$. Write $\Lambda=(\lambda_1,\dots,\lambda_m)$ for the roots of $m_A$ with possible multiplicity, so $m_A(z)=\prod_{k=1}^{m}(z-\lambda_k)$, and set
\[
d_k=\sqrt{1-|\lambda_k|^2}\ \ (>0),
\qquad
\phi(z)=\prod_{k=1}^{m}\frac{\lambda_k-z}{1-\overline{\lambda_k}\,z}.
\]

We denote by \emph{$S(\Lambda)$ the matrix of the compressed shift $S(\phi)$} in the Malmquist-Walsh orthonormal basis. The matrix $S(\Lambda)$ has the upper-triangular form
\begin{equation} \label{eq:livsic}
    \bigl[S(\Lambda)\bigr]_{jk}=
\begin{cases}
\lambda_k, & j=k,\\
\,d_j\Bigl(\prod_{j<\ell<k}(-\overline{\lambda_\ell})\Bigr)d_k,
  & j<k,\\
0, & j>k.
\end{cases}
\end{equation}
The diagonal recovers the eigenvalues of $A$. For a detailed presentation of compressed shifts and model spaces,  we refer to \cite[\S 4,5,9 ]{Garcia2016}.  It can be checked that $S(\Lambda)$ is a contraction and has defect index 1.

\begin{theorem}[Theorem 1.4 \cite{Wu1997}, Lemma 4 \cite{Nakazi95}] \label{thm:wu_inside_Sn}
Let $T = U \oplus A$ be a contraction, with $U$ unitary and $A$ completely non unitary.  Let $u_1,\dots,u_p$ be the distinct eigenvalues of $U$.
Let $\Lambda=(\lambda_1,\dots,\lambda_m)$ be the roots of the minimal polynomial of $A$ with multiplicities.
Then $T$ power dilates to 
\[ 
 B = \bigoplus^N \left( D \oplus S(\Lambda) \right) \text{ where } D = \diag(u_1,\dots,u_p)
\]
with $N = \max( r_U,r_A) \leq n$, where $r_U$ is the maximal eigenvalue multiplicity of the $u_k$ and $r_A$ is the defect index of $A.$
\end{theorem}

Note that $S(\Lambda)$ is not unitary, thus this is not a unitary dilation.  \cref{lem:dilation} gives directly
    \[
    \W \subseteq W \left( D \oplus S(\Lambda)\right).
    \]

\begin{theorem}\label{thm:np1}
Assume (H1)-(H3) and $\rho\succeq0$. Then $\rho$ admits a separable decomposition
\begin{equation*}\label{eq:decompnp1}
  \rho=\sum_{k=1}^{r}p_k\,\alpha_k\otimes\beta_k
\end{equation*}
with $r \leq n+1$ terms in which every $\alpha_k$ is a pure qubit state and every $\beta_k$ a possibly mixed state on $\C^n$. Let $T$ be the contraction associated with $\rho$, and $m_T$ is minimal polynomial,  then $\sr(\rho)\leq \deg m_T +1\leq n+1$.
\end{theorem}

\begin{proof}
If $T$ is normal or unitary, \cref{lem:dilation_decomposition} gives a decomposition with at most $n$ terms. 
If $T = U \oplus A$ has defect index $1$ or more, by \cref{thm:wu_inside_Sn}, $T$ dilates to
\[ B = \bigoplus^N( D\oplus S(\Lambda))
\]
a $\deg m_T N \times \deg m_T N$ matrix with $\deg m_T =\deg m_U + \deg m_A =  p+m \leq n$. We construct a unitary dilation similarly to \cite[Thm 2.2]{Wu1997}. Since $S(\Lambda)$ has defect index 1, $S(\Lambda)$ dilates to a unitary $V$ of size $(m+1)$ by \cref{lem:defect_rank}. Thus, $T$ dilates to 
\[ U_{dil} = \bigoplus^N( D\oplus V )
\]
which is unitary of size $(\deg m_T + 1)N \times (\deg m_T+1)N$. By \cref{lem:dilation,thm:unitary_polygon}, $\W$ is inside the polygon formed by the $\deg m_T +1$ distinct eigenvalues of $U_{dil}$.  Let $e^{i\theta_k}$, be its at most $\deg m_T +1 $ distinct eigenvalues.  
Let $J$ be the isometry associated to the dilation $U_{dil}$. Thanks to \cref{lem:dilation_decomposition}, the operators
\[
L_ k  = J^\dagger\,E_ k  J 
\]
are hermitian, positive, and satisfy 
\begin{equation}\label{eq:dil_moments}
    \sum_{ k =1}^{\deg m_T+1} L_ k  = \Id{n},\quad
\sum_{ k =1}^{\deg m_T+1} \cos(\theta_ k ) L_ k  = C_1,\quad
\sum_{ k =1}^{\deg m_T+1} \sin(\theta_ k ) L_ k  = C_3.
\end{equation}

For each $ k $ with $L_k\neq 0$, define
\begin{equation}
  p_ k =\frac{\Tr(L_ k )}{n},
  \quad
\alpha_ k =\tfrac12\bigl(\Id2+\cos\theta_k\,X+\sin\theta_k\,Z\bigr),
  \quad
  \beta_ k =\frac{L_ k }{n p_ k }.
\end{equation}
With $\gamma_ k =p_ k \beta_ k $, dividing \eqref{eq:dil_moments} by $n$ reproduces exactly the moment conditions \eqref{eq:moments} with $v_k=(\cos\theta_k,\sin\theta_k)$. Hence $\rho=\sum_kp_k\,\alpha_k\otimes\beta_k$ with $p+m+1 \leq n+1$ terms. 
The weights satisfy $p_k\ge0$ and $\sum_kp_k=\Tr(\Id n)/n=1$ by the first identity in \eqref{eq:dil_moments}. Each $\alpha_k$ is a pure qubit state (Bloch vector $(\cos\theta_k,0,\sin\theta_k)$); each $\beta_k$ is a possibly mixed state.
\end{proof}

\begin{figure}
  \centering
  \begin{subfigure}[b]{0.32\textwidth}
    \centering
    \includegraphics[width=\linewidth, trim={0 10 0 0},clip ]{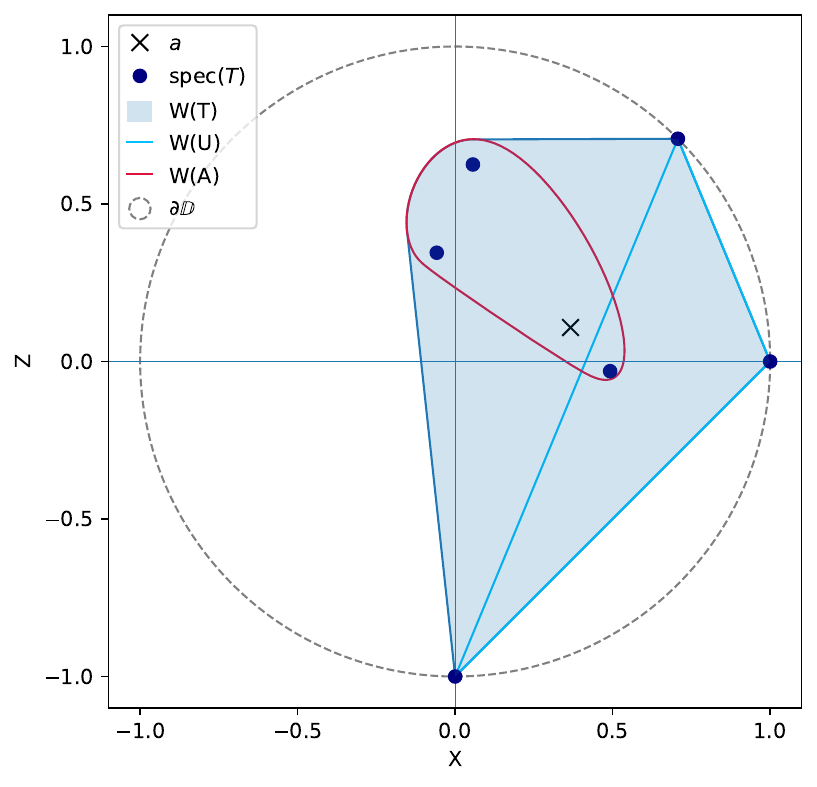}
    \caption{}
    \label{fig:start}
  \end{subfigure}
  \hfill
  \begin{subfigure}[b]{0.32\textwidth}
    \centering
    \includegraphics[width=\linewidth, trim={0 10 0 0},clip ]{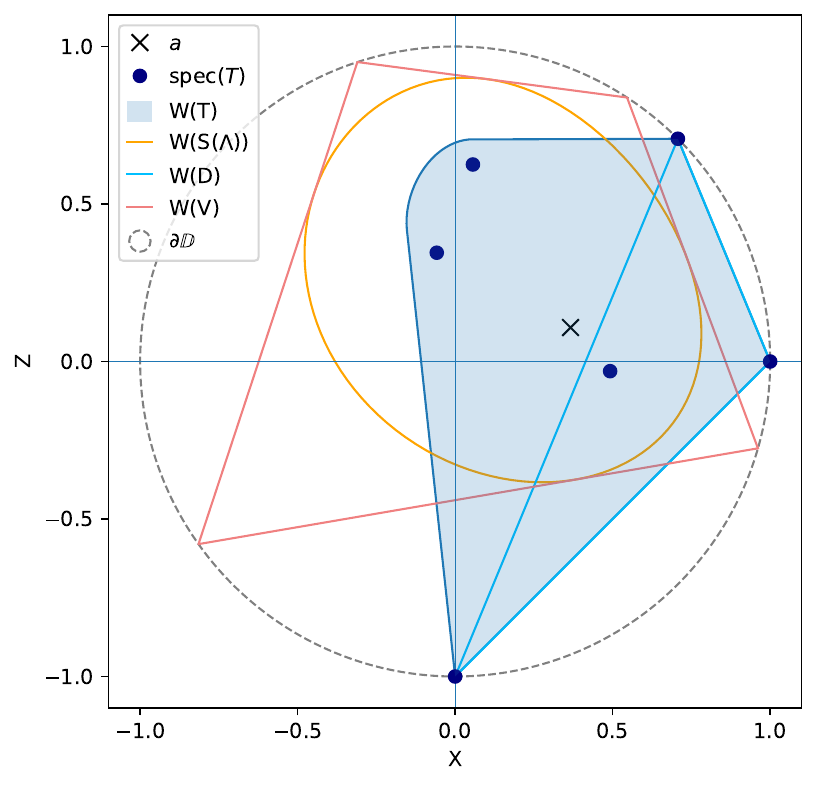}
    \caption{}
    \label{fig:int}
  \end{subfigure}
  \hfill
    \begin{subfigure}[b]{0.32\textwidth}
    \centering
    \includegraphics[width=\linewidth, trim={0 10 0 0},clip ]{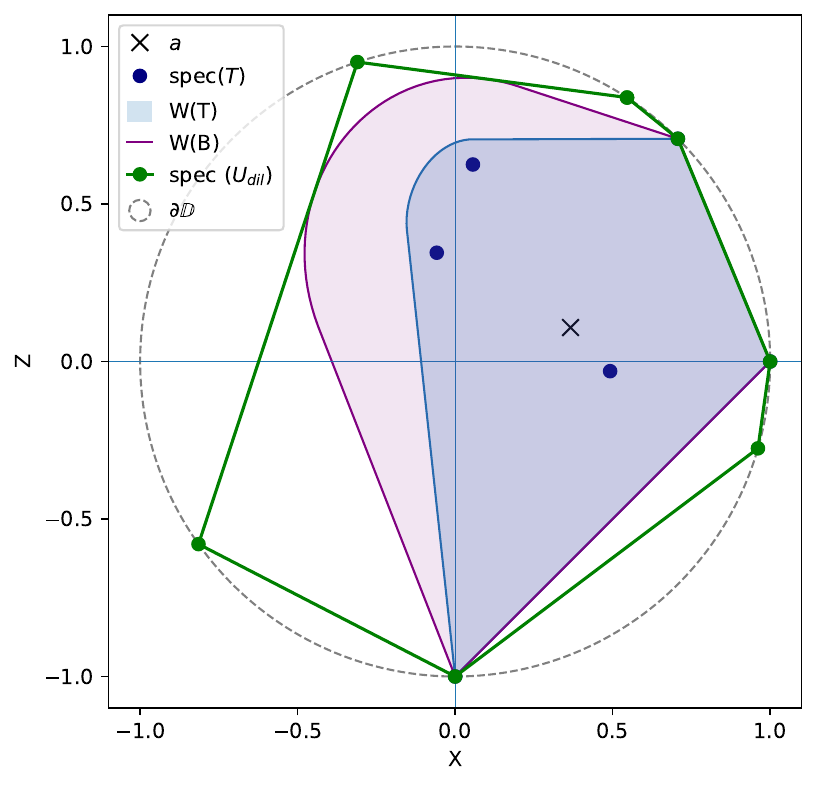}
    \caption{}
    \label{fig:end}
  \end{subfigure}
  \caption{Construction of  $U_{dil}$ for a bipartite system with $n=6$. (A) $T$ can be decomposed as $T= U \oplus A$, $U\in U(3)$ and $ A \in M_3(\C)$ c.n.u., thus $W(T) = \conv (W(U) \cup W(A))$. The defect index is $3$. (B) $D$ and $S(\Lambda)$ that compose $B$ are represented. $S(\Lambda)$ dilates to $V$ and $W(D) = W(U)$ because $D$ is unitarily equivalent to $U$. (C) $W(B) = \conv (W(D) \cup W(S(\Lambda))) $ and $W(U_{dil}) = \conv (W(D) \cup W(V))$. Each eigenvalue of $U_{dil}$ has multiplicity $3$.}
  \label{fig:proof}
\end{figure}

The construction of the associated isometry $J$ such that  $T =  J^\dagger U_{dil} J$ is given in \cref{app:isometry}.

The different steps of the construction are illustrated in \cref{fig:proof} and various examples for different $n$ are presented in \cref{fig:mixed}.

\begin{figure}
    \centering
    \includegraphics[width=\linewidth, trim={10 20 10 30},clip ]{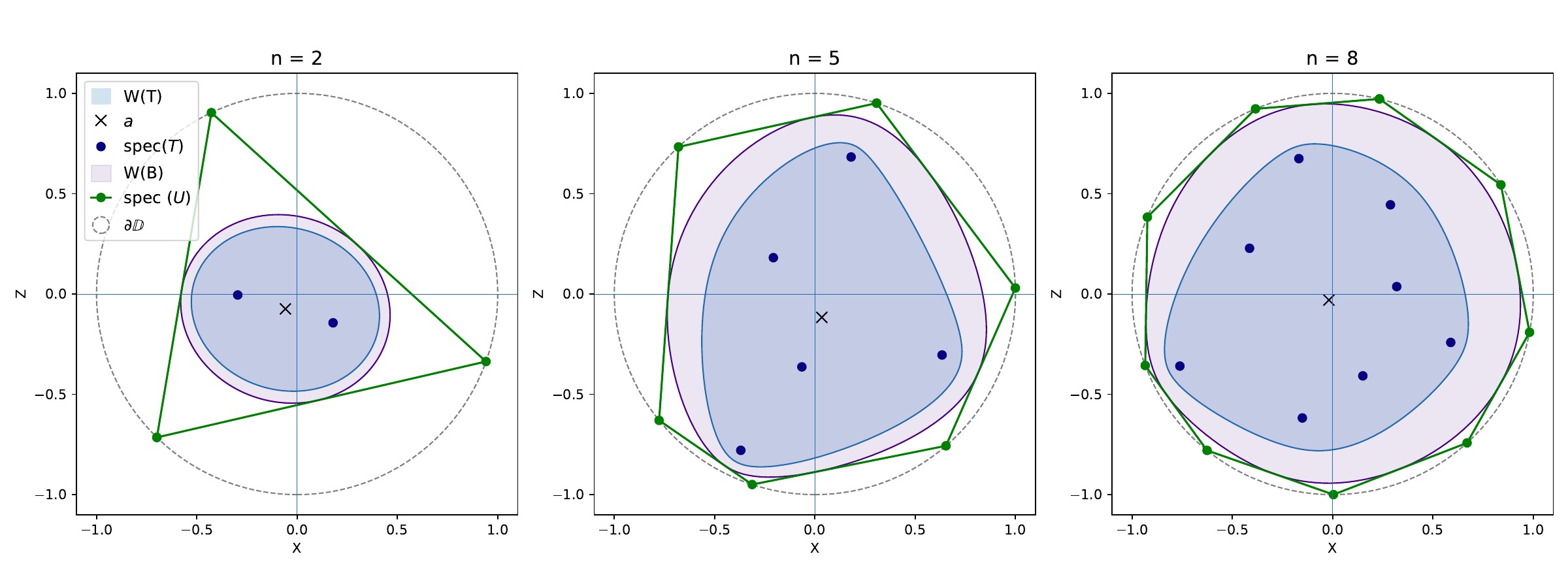}
    \caption{Numerical ranges of $T, B$ and $U_{dil}$ for 3 bipartite systems with $n=2,5,8$. Here $T$ is c.n.u. Both $B$ and $U_{dil}$ are dilations of $T$, but only $U_{dil}$ is normal and yields a decomposition. $W(U_{dil})$ is tangent to $W(B)$. In each case, $\rk(\rho) =2n$  and each $U_{dil}$ eigenvalue has multiplicity $n$.}
    \label{fig:mixed}s
\end{figure}
\begin{remark}
    There exist states that require at least $n+1$ terms in their minimal separable decomposition.
\end{remark}

\begin{example}\label{ex:n+1}
    Choose $T$ to be unitarily equivalent to the nilpotent shift 
    \[J_n = \begin{pmatrix}
    0 & 1 & \\ &\ddots & 1 \\  &&0
    \end{pmatrix}.
    \]
    Then set $C_1 = \tfrac{1}{2}(T + T^\dagger)$, $C_3 = \tfrac{1}{2i}(T - T^\dagger)$. The matrices $C_1, C_3$ are both Hermitian. They are independent, $\Tr(C_1^\dagger C_3) = \Tr (C_1)= \Tr (C_3)= 0$.
    The numerical range is the disk $W(T) = W(J_n)= \overline{D}\!\left(0, \cos\tfrac{\pi}{n+1}\right)$ \cite{HaagerupDeLaHarpe1992}. Then $W(T)$ is enclosed by the regular $n+1$-gon and cannot be enclosed by a polygon in the unit disk with fewer than $n+1$ vertices. This is illustrated in \cref{fig:shift}. 
    Since $\norm{T} = 1$, the state $\rho = \frac{1}{2n}(\Id 2 \otimes\Id n + X\otimes C_1 + Z\otimes C_3)$ is positive by \cref{prop:contraction}. Since $C_1,C_3$ are independent, $\OSR(\rho) =3 $.
\end{example}

\begin{figure}
    \centering
    \includegraphics[width=0.38\linewidth, trim={0 12 0 20},clip ]{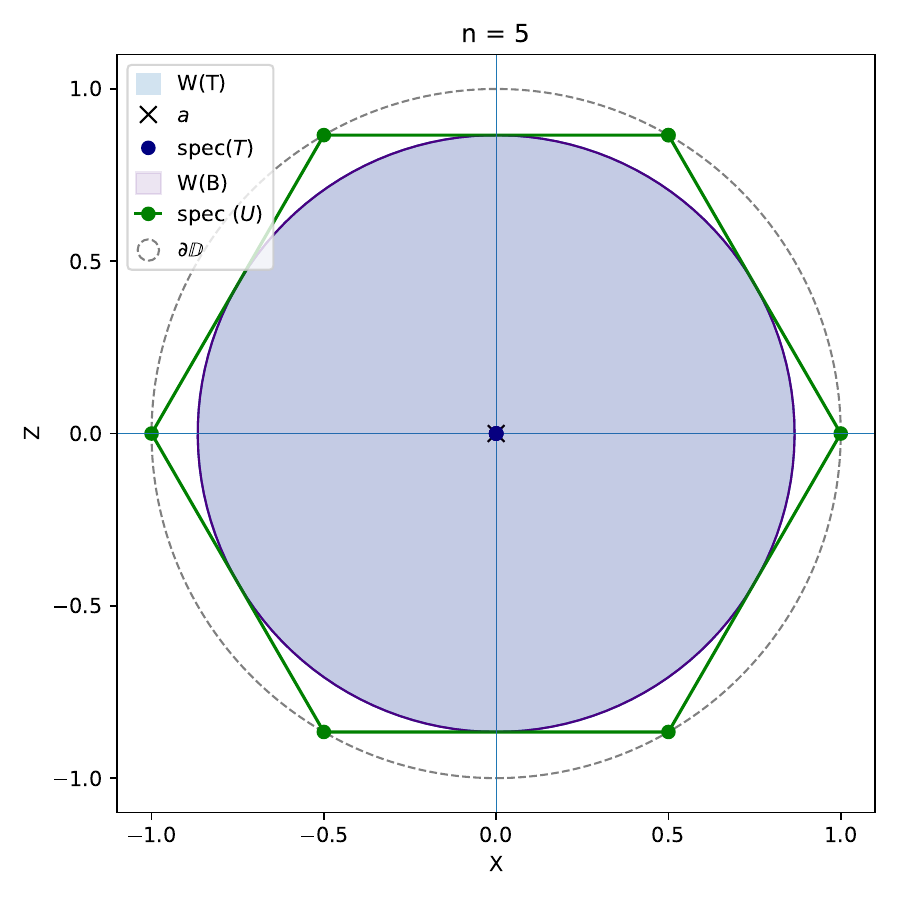}
    \caption{Numerical range of the nilpotent shift $J_5$. All the eigenvalues are $0$.}
    \label{fig:shift}
\end{figure}

This is not a unique example. The class of states associated with contractions $T$ of defect index $1$ and no unimodular eigenvalues requires at least $n+1$ terms (pure and mixed). For such contractions, $W(T)$ is contained in no $m$-gon in $\bD$ for $m \leq n$ \cite[Cor. 7.2.5]{Gau_Wu_2021}. 

If $\rk(\rho) = n$, since the defect index of $T$ is 0, $T$ is unitary, and the (pure) product state decomposition can be directly extracted from $T$ itself. 

\Cref{thm:main_pure} (resp \cref{thm:main_mixed}) follows from \cref{thm:2n} (resp. \cref{thm:np1}), modulo the reduction of a general $\OSR=3$ state to the normal form (H1)-(H3). The filtering preserves the rank and the length of the decomposition, yielding both results.

\section{Summary of the constructions}\label{sec:pseudocode}

We summarize the steps to construct the separable decomposition of a state $\rho$ with $\OSR(\rho) = 3$. 

{\it Step 1: reduction to the normal form.}
\begin{enumerate}
  \item Compute
  \[
    \rho_1 = \frac{(V_1^\dagger \otimes \Id{n})\,\rho\,(V_1 \otimes \Id{n})}
                  {\Tr\bigl[(V_1^\dagger \otimes \Id{n})\,\rho\,(V_1 \otimes \Id{n})\bigr]},
  \]
  where $V_1$ is given by the Cariello SLOCC filtering (see \cref{app:cariello_construction}). Then $\rho_1 = \rho_1^{T_A}$.

  \item If $\rho_{1,B}$ is not full rank, let $U_1$ be a unitary diagonalizing $\rho_{1,B}$, i.e.\ $\rho_{1,B} = U_1\bigl(\Lambda_m \oplus 0_{n-m}\bigr)U_1^\dagger$, and write
  \[
    (\Id{2}\otimes U_1^\dagger)\,\rho_1\,(\Id{2}\otimes U_1)
    = \begin{pmatrix}
        \rho_{2,(1,1)} & 0 & \rho_{2,(1,2)} & 0 \\
        0 & 0 & 0 & 0 \\
        \rho_{2,(2,1)} & 0 & \rho_{2,(2,2)} & 0 \\
        0 & 0 & 0 & 0
      \end{pmatrix}.
  \]
  Then compute
  \[
    \rho_2 = \begin{pmatrix}
               \rho_{2,(1,1)} & \rho_{2,(1,2)} \\
               \rho_{2,(2,1)} & \rho_{2,(2,2)}
             \end{pmatrix}  \in \C^2\otimes \C^m\
  \]
  (see \cref{lem:full_local_rank_reduction}).
  \item Now that $\rho_{2,B}$ has full rank, let $\tau_B$ be the unique positive-definite Hermitian square root of $\rho_{2,B}$, i.e.\ $\tau_B^2 = \rho_{2,B}$. Compute
  \[
    \rho_3 = \frac{(\Id{2} \otimes \tau_B^{-1})\,\rho_2\,(\Id{2} \otimes \tau_B^{-1})}
                  {\Tr\bigl[(\Id{2} \otimes \tau_B^{-1})\,\rho_2\,(\Id{2} \otimes \tau_B^{-1})\bigr]}.
  \]
  Then $\rho_{3,B} = \Id{m}/m$.
\end{enumerate}
Define $C_1 = m\Tr_A\bigl[\rho_3\,(X \otimes \Id{m})\bigr]$ and $C_3 = m\Tr_A\bigl[\rho_3\,(Z \otimes \Id{m})\bigr]$. We obtain the normal form
\[
  \rho_3 = \frac{1}{2m}\bigl(\Id{2}\otimes\Id{m} + X\otimes C_1 + Z\otimes C_3\bigr).
\]
Let $T = C_1 + iC_3$ be the associated contraction, and $d = \rk\,(I - T^\dagger T)^{1/2}$ its defect index.

{\it Step 2A: dilation for the pure-state decomposition.}
Compute the dilation $U_2$ of $T$ on $\mathbb{C}^m \oplus \mathbb{C}^d$ given by \cref{lem:defect_rank}, and set $J_2 = \begin{psmallmatrix} \Id m \\ 0_d \end{psmallmatrix}$ for the isometry.

{\it Step 2B: dilation for the mixed-state decomposition.}
Compute the dilation $U_3$ of $T$ given by \cref{thm:np1}, and the isometry $J_3$ given in \cref{app:isometry}.

{\it Step 3: the decomposition.}
Fix $j \in \{2,3\}$. Let $e^{i\theta_1},\dots,e^{i\theta_r}$ be the $r$ eigenvalues of $U_j$. The eigenvalues are taken with their multiplicity for the decomposition with pure product states and distinct for the decomposition with mixed product states. Let $E_k$ be the orthogonal projection onto the eigenspace of $e^{i\theta_k}$, $1 \le k \le r$. The projector $E_k$ has $\rk 1$ for the decomposition with pure product states. Define
\begin{equation*}
    L_k = J_j^\dagger E_k J_j .
\end{equation*}
For each $k$ with $L_k \neq 0$, define
\begin{equation*}
  p_k = \frac{\Tr(L_k)}{m}, \qquad
  \alpha_k = \tfrac12\bigl(\Id{2} + \cos\theta_k\,X + \sin\theta_k\,Z\bigr), \qquad
  \beta_k = \frac{L_k}{m\,p_k}.
\end{equation*}
The decomposition of $\rho_3$ is then $\rho_3 = \sum_k p_k\,\alpha_k \otimes \beta_k$.

{\it Step 4: invert Step 1.}
\begin{enumerate}
  \item Compute  $\gamma_k = \frac{\tau_B \beta_k  \tau_B}{\Tr \left(\tau_B \beta_k  \tau_B \right)}$ and $p'_k = p_k\frac{\Tr \left[\tau_B \beta_k  \tau_B \right]}{\Tr\bigl[(\Id{2} \otimes \tau_B)\,\rho_3\,(\Id{2} \otimes \tau_B)\bigr]}$. Then $\rho_2 =  \sum_k p'_k \alpha_k \otimes \gamma_k $.
  \item If $\rho_{1,B}$ was not full rank, pad the $\gamma_k$ with $0$ matrices to obtain an operator on $\C^n$, and compute $\zeta_k$
    \[
    \zeta_k 
    = U_1 \begin{pmatrix}
        \gamma_k & 0 \\
        0 & 0_{n-m} &\\
      \end{pmatrix}  U_1^{\dagger}.
  \] 
  Then $\rho_1 = \sum_k p'_k \alpha_k\otimes \zeta_k$.
  \item Compute $\xi_k = \frac{V_1^{-\dagger } \alpha_k  V_1^{-1}}{\Tr [V_1^{-\dagger } \alpha_k  V_1^{-1}]}$, and  $p''_k = p'_k\frac{\Tr [V_1^{-\dagger } \alpha_k  V_1^{-1}]}{\Tr\bigl[(V_1^{-\dagger } \otimes \Id{n})\,\rho_1\,(V_1^{-1} \otimes \Id{n})\bigr]]}$. Then the final decomposition is 
  \[
  \rho =  \sum_k p''_k \xi_k \otimes \zeta_k .
  \]
\end{enumerate}

\bigskip
\section{Discussions and Remarks}\label{sec:n2}
\bigskip

Most statements below are known, either in operator theory (dilations, numerical ranges, matrix convex sets) or in quantum information (entanglement-breaking channels, steering ellipsoids). Our purpose is to highlight the link between these results and the construction of separable decompositions, and their geometric interpretation.

\subsection{The triangle case}

By \cref{prop:geom}, a necessary condition for states to be separable is to be enclosed in a polygon inside the unit disk.  In \cref{thm:2n,thm:np1}, the dilation construction provides both the construction of the vertices and the separable decomposition. When a triangle suffices to enclose  $\W$, there is a direct construction of the decomposition from the vertices alone \cite[Prop. 2.5]{Wu1997}.

\subsubsection{Construction}

Throughout this section, fix three non-collinear points $v_k =(x_k, z_k )\in\Disk$, $k =1,2,3$,  the vertices of the triangle, and let $\Delta=\conv(v_1,v_2,v_3)$ for the triangle.

\begin{proposition}[Barycentric coordinates]\label{prop:bary}
There exist unique affine functions $\ell_k (w)=b^{0}_k +b^{1}_k  w_1+b^{3}_k  w_3$ with 
\begin{eqnarray}
    &\ell_k (v_j)=\delta_{k j}\label{eq:dirac} \\
    &\sum_k \ell_k (w)=1,\qquad \sum_k  v_k \,\ell_k (w)=w
  \qquad\forall w\in\R^2,\label{eq:partition}
\end{eqnarray}
and $\Delta=\{w:\ell_k (w)\ge0,\ k =1,2,3\}$. Their coefficient vectors $(b^{0}_k,b^{1}_k,b^{3}_k)^T$ are the columns of $V^{-1}$, with $V$ given by
\[V=\begin{pmatrix}1&x_1&z_1\\1&x_2&z_2\\1&x_3&z_3\end{pmatrix}.\]
\end{proposition}

\begin{proof}
Condition \eqref{eq:dirac} and uniqueness follow from $V(b^{0}_k ,b^{1}_k ,b^{3}_k )^T=e_k $ and $V$ is invertible since $\det V=\pm2\,\mathrm{Area}(\Delta)\neq0$. Both sides of each identity in \eqref{eq:partition} are affine in $w$ and agree at the three non-collinear points $v_j$, and hence everywhere. For the last property, note that $w \in \Delta$ is in the convex hull of $v_1,v_2,v_3$, thus there exist $c_k \geq 0 $ such that  $\sum c_k = 1$, and $w = \sum c_k v_k$. By uniqueness, $c_k = \ell_k(w)$ and $\ell_k(w)\geq 0$. Conversely, if $\ell_k(w)\geq 0$, by \eqref{eq:partition} $w \in \Delta$.
\end{proof}

For $\ell_k(w)=b_k^{0}+b_k^{1}w_1+b_k^{3}w_3$, define $\ell_k(C_1,C_3)=b_k^{0}\Id n+b_k^{1}C_1+b_k^{3}C_3$.

\begin{lemma}\label{lem:nr}
$\ell(C_1,C_3)\succeq0\iff\ell\ge0$ on $\W$. Hence, for $L_k =\ell_k (C_1,C_3)$, we have $\W\subseteq\Delta\iff L_k \succeq0$ for $k =1,2,3$.
\end{lemma}

\begin{proof}
For every unit vector $\ket{\psi}$, $\braket{\psi|\ell(C_1,C_3)|\psi}=\ell(\braket{\psi|C_1|\psi},\braket{\psi|C_3|\psi})$ which directly gives the first equivalence. The second claim follows from \cref{prop:bary}.
\end{proof}

\begin{proposition}\label{prop:threeterm}
Assume (H1)-(H3) and $\W\subseteq\Delta$. Then $(L_k )_{k =1}^3$ satisfies
$\sum_k  v_k  L_k = (C_1,C_3)$, and 
\[
  \rho=\sum_{k =1}^3 p_k \,\alpha_k \otimes\beta_k ,
  \quad\text{ with } p_k =\ell_k (\vect a), \, \alpha_k = \frac{1}{2}(I + x_kX + z_k Z), \, \beta_k =\frac{L_k }{n\,p_k }.
\]
\end{proposition}

\begin{proof}
\cref{eq:partition} gives the moment conditions $\sum_k  L_k =\Id n$, $\sum_k  v_k  L_k = (C_1,C_3)$. Since $\Tr C_j=na_j$, $\Tr L_k =n\,\ell_k (\vect a) = n p_k \geq 0$ and $\sum_k p_k = 1$ by \eqref{eq:partition}. $\alpha_k \succeq 0 $ because $v_k \in \Disk$. By \cref{lem:nr} $L_k \succeq0$ and thus $\beta_k \succeq0$.
\end{proof}

\begin{remark}
With this construction, we can construct the isometry $J=\sum_k \ket{k} \otimes L_k^{1/2}$ that maps $\C^n\to\C^3\otimes\C^n$, $J^\dagger J=\sum_k L_k=\Id n$, with $P_k=\ketbra{k}{k}\otimes\Id n$ and $N=\sum_k v_k    P_k$. Then, $T$ dilates to the normal operator $N$, $J^\dagger NJ=\sum_k v_k L_k=T$.
\end{remark}

This construction does not extend to polygons with more than three vertices. The barycentric coordinates are uniquely defined affine functions only for a simplex. In fact, simplices are the only polytopes $P$ for which $W(T)\subseteq P$ always implies the existence of a normal dilation of $T$ with spectrum in $P$ \cite[Thm. 4.1]{PasserShalitSolel2018}.

The same construction in $\R^3$, with a tetrahedron, can be used to treat the case $\OSR(\rho)=4$ for $n=2$.  By filtering a two-qubit state, we can always consider $\rho$ with a maximally mixed reduced state $\rho_B=\Tr_A\rho$, such that
\begin{equation}\label{eq:rho3}
  \rho=\tfrac14\left(I\otimes I+X\otimes C_1+Y\otimes C_2+Z\otimes C_3\right),
\end{equation}
with $C_1,C_2,C_3$ Hermitian. The convex hull $\conv W(C_1,C_2,C_3)$ of their joint numerical range coincides with the steering ellipsoid of $\rho$. \cref{prop:bary,lem:nr,prop:threeterm} hold for a tetrahedron in $\R^3$ (four barycentric coordinates, $V\in M_4$) and give four-term separable decompositions. This recovers the sufficiency part of the nested tetrahedron condition of Jevtic, Pusey, Jennings and Rudolph. Necessity is given by the generalization of \cref{prop:geom}.

\begin{theorem}[Nested tetrahedron condition \cite{Jevtic2014}]\label{thm:jevtic}
A two-qubit state $\rho$ is separable if and only if its steering ellipsoid is contained in a tetrahedron which is itself contained in the unit ball of $\R^3$.
\end{theorem}

\subsubsection{Two-qubit special case ($n=2$)}
By \cref{prop:threeterm}, any triangle $\Delta \subset \Disk$  with 
$\W \subseteq \Delta $ yields a separable three-term decomposition. For $n=2$, $\partial \W$ is an ellipse $\mathcal{E}$ \cite{Toeplitz1918,Li1996} and  such a triangle always exists. 

\begin{theorem}\label{thm:n2}
For $n=2$,  any state $\rho$ with $\OSR(\rho) =3$ admits a separable decomposition with at most $3$ terms.
\end{theorem}

\begin{proof}
This proof can be done without relying on the heavy construction of \cref{thm:np1}. Reduce $\rho$ to the normal form \cref{eq:rhoC}. By \cref{prop:contraction}, $r=\norm T\le1$, and $r > 0$ otherwise $\rho = \tfrac{1}{4}\Id 4$. Define $S=T/r$, such that $\norm S=1$ and the singular values of $S$ are $s_1 = 1\ge s_2$. Thus, $S$ has defect index at most $1$ and by \cref{lem:defect_rank}, there exist a unitary $U$ and an isometry $J$ such that 
\[
T = rS = r J^\dagger U J
\]
Then the vertices given by $r \lambda_k$, where $\lambda_k$ are the eigenvalues of $U$ form a triangle enclosing $\W$. The triangle vertices lie on the circle of radius $r$ centered at $0$ (see \cref{fig:TS}). These three points are non-collinear, otherwise $\W$ would be a segment and $\OSR(\rho)\leq2$. Thus, the decomposition can be obtained by \cref{prop:threeterm}. The construction is illustrated in \cref{fig:TS}.
\end{proof}

When $\W$ is not degenerate the set of minimal separable decompositions in $xz$-form corresponds to the set of triangles $\Delta \subset \Disk$ such that $ \W \subset \Delta$.

\subsection{Poncelet property}

The Poncelet theorem gives another way to construct the enclosing triangle for $n=2$.
It states that if two conics, here a circle $\mathcal{C}$ and an ellipse $\mathcal{E}$, admit one $m$-sided polygon inscribed in $\mathcal{C}$ and circumscribed around $\mathcal{E}$, then they admit infinitely many such polygons, and one vertex may be chosen on $\mathcal{C}$ arbitrarily (see \cref{fig:Poncelet}). 

The existence of the polygon, which in our case is a triangle, is given in the proof of \cref{thm:n2}.  One can choose for  $\mathcal{C}$ the circle centered at $0$ of radius $\norm{T}$, the ellipse $\mathcal{E}$ is the boundary of $\W$. The tangency of the triangle to $\mathcal{E}$ was proven in \cite{Gau1998, Mirman1998, Daepp2002}. Thus, by Poncelet's theorem, we can conclude that there exist infinitely many triangles obtained as follows: pick any $v_1\in\mathcal C$, choose one of the two tangents to $\mathcal E$ through $v_1$, and let $v_2$ be its second intersection with $\mathcal C$, then take the other tangent to $\mathcal E$ through $v_2$ to obtain $v_3$. The Poncelet theorem guarantees that the tangent from $v_3$ closes the triangle at $v_1$. The link between numerical ranges and Poncelet's property for $n=2$ was extensively studied in \cite{DaeppGorkinShafferVoss2018}. 

The Poncelet property has been generalized to an entire class of contractions. A contraction $A$ of dimension $n$ is in the class $\mathcal{S}_n$ if the defect index of $A$ is 1  and $\operatorname{spec}(A) \subseteq \oDisk$. The class $\mathcal{S}_n$ has the $n+1$-Poncelet property \cite{Gau1998,GauWu2003, Mirman1998}. For every $v\in\bD$, there is a unique $(n+1)$-gon inscribed in $\bD$, circumscribed about $\partial W(A)$, and having $v$ as a vertex. This polygon is in one-to-one correspondence (up to unitary similarity) with the unitary dilation of $A$ of $\dim n +1 $, the vertices being the eigenvalues of the unitary. 

\begin{proposition}
    Let $T$ be the contraction associated with a state $\rho$ verifying H1-H3. If $T$ has one unimodular eigenvalue $\lambda_i$, then $\sr \leq n $. 
\end{proposition}

\begin{proof}[Sketch of the proof]
    Since in \cref{thm:np1}, $S(\Lambda) \in \mathcal{S}_m$, $W(S(\Lambda))$ has the $m+1$-Poncelet property, thus the dilation $V$ of $S(\Lambda)$ can be chosen such that one of its eigenvalues coincides with $\lambda_i$ reducing the number of terms in the decomposition by one.
\end{proof}

\begin{figure}
  \centering
  \begin{subfigure}[b]{0.48\textwidth}
    \centering
    \includegraphics[width=0.7\linewidth, trim={0 0 0 0},clip ]{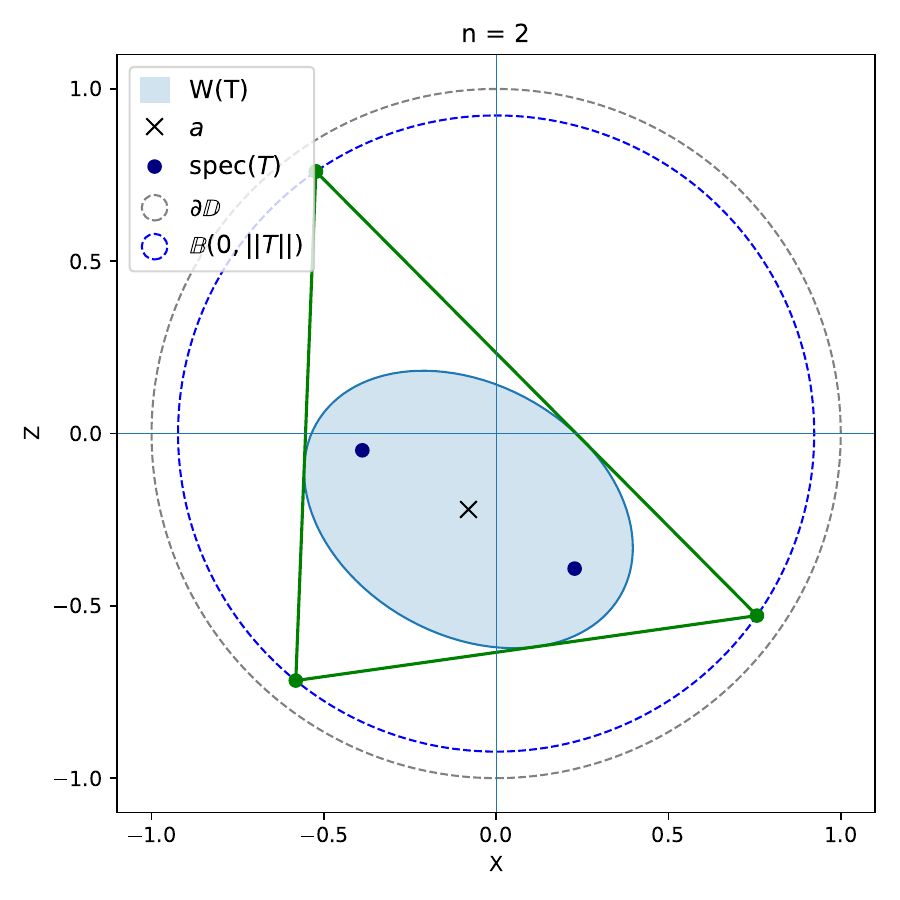}
    \caption{$n=2$ ellipse}
    \label{fig:TS}
  \end{subfigure}
  \hfill
  \begin{subfigure}[b]{0.48\textwidth}
    \centering
    \includegraphics[width=0.7\linewidth, trim={0 0 0 0},clip ]{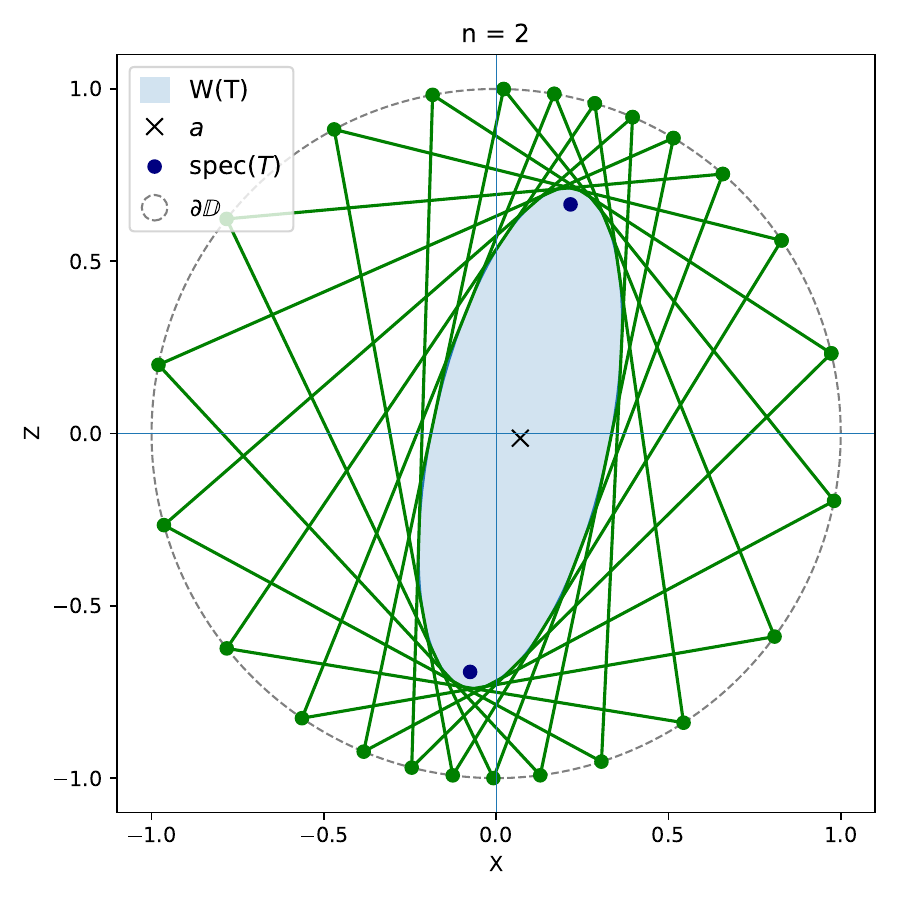}
    \caption{Poncelet property}
    \label{fig:Poncelet}
  \end{subfigure}
  \caption{Contraction $T$ with $n=2$. (A) The vertices of the triangle circumscribing the numerical range $W(T)$ are on the circle of radius $\norm{T}$. (B) Any point on the circle $\mathcal{C}$ generates a circumscribing triangle.}
  \label{fig:n=2}
\end{figure}

\subsection{Dilations and separability}

Let $\rho$ be in the normal form~\eqref{eq:rhoC}. The construction used in \cref{thm:2n,thm:np1} shows that any unitary dilation of $T$ yields a separable decomposition with pure states on $\Hs_A$. Thus, minimizing over the unitary dilations $U$ of $T$ yields  $\spr \leq \min_U{\dim U}$ and $\sr \leq \min_U{|\text{distinct spec}(U)|}$ where $|\text{distinct spec}(U)|$ is the number of distinct eigenvalues of $U$. 

Conversely, given $\rho$ in normal form and a separable decomposition $\rho=\sum_{k=1}^{r}p_k\,\alpha_k\otimes\beta_k$ with pure $\alpha_k=\tfrac12(I_2+\cos\theta_k X+\sin\theta_k Z)$, we can construct a unitary dilation of $T$. One proceeds as follows. Take $L_k =n\,p_k\beta_k\succeq0$. The moment conditions~\eqref{eq:moments} read $\sum_kL_k=I_n$ and $\sum_ke^{i\theta_k}L_k=T$. Thus, $\{L_k\}$ is a POVM on $\mathcal H_B$.
Let us take Naimark's isometry $J_N=\sum_k\sqrt{L_k}\otimes|k\rangle$ and unitary $U_N =\sum_ke^{i\theta_k}\,\Pi_k\otimes|k\rangle\langle k|$  where $\Pi_k$ is the projector on $\ran L_k$. They satisfy $J_N^\dagger U_NJ_N=T$. Hence, we have shown that every separable decomposition yields a Naimark dilation of $T$. Therefore, for a {\it minimal} separable decomposition with pure states, we have $\dim U_N = \spr(\rho)$ and for a {\it minimal} separable decomposition with mixed states, we have $|\text{distinct spec}(U_N)| = \sr(\rho)$. 

At this point note that the last identity does not give us an operational way to compute $\sr(\rho)$ because a minimal decomposition is needed, but the constructions in this paper gives us (in general) only decompositions with length at most $n+1$.

Combining the previous remarks we see that we in fact have the equalities:
\begin{equation}\label{eq:ranks-dilation}
  \operatorname{p-sep-rank}(\rho)=\min_U\dim U,
  \qquad
  \operatorname{sep-rank}(\rho)=\min_U|\text{distinct spec}(U)|,
\end{equation}
where $U$ ranges over the unitary dilations of $T$. The first identity recovers a known result, $\spr(\rho) = \rk(\rho)$ thanks to \cref{thm:2n} and \cite{Kraus2000}. The second identity suggests that a closer study of the unitary dilations of $T$ may give further insight into the determination of separable rank, beyond the upper bound.

\subsection{Entanglement-breaking channels}
Consider the quantum channel $\Phi:M_n(\mathbb C)\to M_2(\mathbb C)$, $\Phi(X)=n\operatorname{Tr}_B[\rho\,(I_2\otimes X^{T})]$, with $\rho$ in normal form and $X^T$ the transpose in the computational basis. Its Choi matrix is $n\rho$, and it is trace preserving since $\rho_B=I_n/n$. By \cite{Horodecki_2003}, separability of $\rho$ is equivalent to $\Phi$ being entanglement breaking. Our decomposition of $\rho$ gives a 'measure-and-prepare' form for the entanglement breaking channel
\begin{equation*}
  \Phi(X)=\sum_{k=1}^{r}\operatorname{Tr}\big(L_k^{T}X\big)\,\alpha_k .
\end{equation*}
The channel measures the POVM $\{L_k^{T}\}$, attached to the dilation $U^T$ of $T^T$. It then prepares the pure state $\alpha_k$.
In \eqref{eq:ranks-dilation}, $\spr(\rho)=\rk\rho$ is the Kraus rank of $\Phi$, since the Kraus rank equals the rank of the Choi matrix $n\rho$, and $\sr(\rho)$ is the minimal number of classical outcomes of $\Phi$, which is at most $n+1$ by \cref{thm:np1}. The classical register between measurement and preparation can therefore be reduced to $\lceil\log_2(n+1)\rceil$ bits, compared with $\lceil\log_2\rk\rho\rceil$ for a Kraus decomposition, which is an improvement of at most 1 bit.

\section*{Code Availability}
The code used to produce the figures and implement the decompositions in this work is available at \url{https://github.com/vantalon/osr3_separability}.

\newpage
\appendix
\crefalias{section}{appendix}
\label{app}

\section{Operator Schmidt decompositions}\label{app:schmidt}

\begin{lemma}\label{lem:herm_pos_first_factor}
Let $ \rho \in M_2(\C)  \otimes  M_n(\C), \rho\succeq 0$ be Hermitian with Hermitian operator Schmidt rank $r$. Then $\rho$ admits a decomposition
\[
  \rho=\sum_{i=1}^r s_i\,A'_i\otimes B'_i,
\]
with $(A'_i)$ and $(B'_i)$ Hermitian, $(A'_i)$ Hilbert-Schmidt orthonormal, and real coefficients $s_i$, such that the first factors satisfy $A'_1\succeq 0$ and $s_1 B'_1\succeq 0$. If $\OSR (\rho) \geq 2$, $A'_1\succ 0$.
\end{lemma}

\begin{proof}
Let $\sigma_m$ be real and $A_m, B_m$ be Hermitian operators such that
\[
  \rho=\sum_{m=1}^r\sigma_m\,A_m\otimes B_m.
\]
This decomposition can be obtained through the Hermitian Schmidt decomposition (using SVD of the matrix $M_{ij}=\Tr\bigl((\Lambda^A_i\otimes\Lambda^B_j)\rho\bigr)$). In that case, it will have the additional properties that  $(A_m)$ and $(B_m)$ are Hilbert-Schmidt orthonormal ($\Tr(A_mA_n)=\delta_{mn}$, as well as for $B$) and $\sigma_m>0$.
Since $\rho\succeq 0$, the reduced operator
\[
\rho_A = \Tr_B(\rho) = \sum_m \sigma_m \Tr(B_m)\,A_m \;\succeq\; 0
\]
lies in $\Span(A_1,\dots,A_r)$. Set
\[
A'_1 = \frac{\rho_A}{\sqrt{\Tr \rho_A^2}} \succeq 0,
\]
and extend it to a Hilbert-Schmidt orthonormal basis $(A'_i)_{1\le i\le r}$ of $\Span(A_1,\dots,A_r)$ by Gram-Schmidt. Each $A'_i$ is Hermitian (a real combination of Hermitian matrices). If $\OSR (\rho) \geq 2$, $\rho_A$ is a mixed state (Lem. \ref{lem:osr3_mixed}) and $\rho_A \succ 0 $, thus $A'_1\succ0$.
Rewriting each $A_m$ in this basis and substituting into $\rho$, we obtain operators $B'_i$ (real combinations of the $B_m$, hence Hermitian) and real coefficients $s_i$ with
\[
  \rho=\sum_{i=1}^r s_i\,A'_i\otimes B'_i,
\]
where $(A'_i)$ are Hilbert-Schmidt orthonormal. 
Finally, by orthonormality $\Tr(A'_iA'_1)=\delta_{i1}$,
\[
\Tr_A\!\big(\rho\,(A'_1\otimes I)\big)
  = \sum_i s_i\,\Tr(A'_iA'_1)\,B'_i = s_1 B'_1.
\]
The left-hand side equals $\Tr_A\!\big((A_1'^{1/2}\otimes I)\,\rho\,(A_1'^{1/2}\otimes I)\big)\succeq 0$, since $\rho\succeq 0$, $A'_1\succeq 0$, and the partial trace preserves positivity. Hence $s_1 B'_1\succeq 0$.
\end{proof}

\section{Cariello construction}\label{app:cariello_construction}
\begin{theorem}[Theorem 3.2 of \cite{cariello2014}]
    Let $\rho\in \Hs$ be a density matrix with $\OSR = 3$, then there exist $V$ invertible such that 
    \begin{equation}
        \Phi_V(\rho) = \Phi_V(\rho)^{T_A}
    \end{equation}
    where $T_A$ indicates the partial transposition of the system A. 
\end{theorem}

\begin{proof}
We reproduce the proof of Cariello \cite{cariello2014}.
It is possible to find a Hermitian decomposition $\rho = \sum_{i=1}^{3} A_i \otimes B_i$ such that $A_1$ is positive definite and $B_1$ is positive semi-definite (Lemma \ref{lem:herm_pos_first_factor}).  Write $A_1 = R R^\dagger$ with $R$ invertible. Let
\[
C = (R^{-1} \otimes I_n)\, \rho \, ((R^{-1})^\dagger \otimes I_n) = I_2 \otimes B_1 + \sum_{i=2}^{3} A_i' \otimes B_i, \qquad A_i' = R^{-1} A_i (R^{-1})^\dagger,\ i = 2,3.
\]
Since $A_2'$ is a Hermitian matrix, there is a unitary matrix $U$ and a real diagonal matrix $D$ such that $A_2' = U D U^\dagger$. Notice that $D \neq \lambda\, I_2$, otherwise $A_2' = \lambda\, Id$ and $A_2 = \lambda A_1$, which is not possible since  $\OSR (\rho)=3$.

Let
\[
E = (U^\dagger \otimes I_n)\, C\, (U \otimes I_n) = I_2 \otimes B_1 + D \otimes B_2 + A_3'' \otimes B_3, \qquad A_3'' = U^\dagger A_3' U.
\]
Since $D \neq \lambda\, I_2$, any diagonal matrix in $M_2(\C)$ can be written as a linear combination of $D$ and $I_2$. Let $D'$ be the diagonal of $A_3''$. Notice that $D'$ is a real diagonal matrix, since $A_3''$ is Hermitian. Write $D' = a\, I_2 + c\, D$, where $a, c$ are real numbers. Thus,
\[
E = I_2 \otimes (B_1 + a B_3) + D \otimes (B_2 + c B_3) + A_3''' \otimes B_3, \qquad
A_3''' = A_3'' - D' = \begin{pmatrix} 0 & \bar{b} \\ b & 0 \end{pmatrix}.
\]
Notice that $b \neq 0$, otherwise $E$ would have $\OSR=2$ and $\rho$ would have $ \OSR= 2$. Let $ S = \diag (1 ,\bar{b})$ and consider $ F = (S \otimes I_n)\, E\, (S^\dagger \otimes I_n).$
Notice that $S S^\dagger$ and $S D S^\dagger$ are diagonal matrices and
\[
S A_3''' S^\dagger = \begin{pmatrix} 0 & \overline{b} b \\ \overline{b} b & 0 \end{pmatrix}
\]
is symmetric too. Thus, we may choose $V = (R^{-1})^\dagger U S^\dagger \otimes \Id{n}$ and $\Phi_V(\rho) = \frac{F}{\Tr F}$ . This is a well define positive semi definite Hermitian matrix invariant under the left partial transposition. 

\end{proof}

\section{Isometry of Wu dilation} 
\label{app:isometry}
Assume $T=U\oplus A$ on $\Hs=\Hs_0\oplus \Hs_1$, with $U$ unitary and $A$ completely non unitary. Let $u_1,\dots,u_p$ be the distinct eigenvalues of $U$, with multiplicities $m_k=\dim\ker(U-u_kI)$ and $r_U=\max_k m_k$. 
Set $N=\max(r_U,r_A)\le\dim H$. Let $m_A(z)=\prod_{j=1}^{m}(z-\lambda_j)$ be the minimal polynomial of $A$, $d_j=\sqrt{1-|\lambda_j|^{2}}$, $b_\lambda(z)=\tfrac{\lambda-z}{1-\overline{\lambda}z}$, and 
\[\varphi=\prod_{j=1}^{m}\frac{\lambda_j-z}{1-\overline{\lambda_j}z} =\prod_{j=1}^{m}b_{\lambda_j}.
\] 
The contraction $T$ dilates to 
\[
B=\bigoplus^{N}\bigl(\diag(u_1,\dots,u_p)\oplus S(\Lambda)\bigr),
\]
where $S(\Lambda)$ is the matrix of the compressed shift $S(\varphi)$ in the Malmquist-Walsh orthonormal basis. We want to construct the isometry $J$ with $T^{k}=J^\dagger  B^{k}J$ for all $k\ge0$.

\subsection{Dilation of U}

Keep one copy of each distinct eigenvalue
\[
D=\diag(u_1,\dots,u_p)\ \text{ on }\ \C^{p}.
\]
Grouping $D\oplus D\oplus\cdots$ by eigenvalue turns the amplified operator into $u_1I\oplus\cdots\oplus u_pI$. Since each eigenspace embeds isometrically into the corresponding coordinate, there is an isometry $V_U:\Hs_0\to\bigoplus^{r_U}\C^{p}$ such that for every $k$,
\[
U^k =V_U^\dagger \left(\bigoplus^{r_U} D \right)^kV_U.
\]

\subsection{Dilation of A}

\subsubsection{Definitions}

Let $H^2$ be the \emph{Hardy space} of the open unit disk $\oDisk$,
\[
H^2=\Bigl\{f(z)=\sum_{k\ge0}a_kz^{k}:a_k\in\C,\ \|f\|^2:=\sum_{k\ge0}|a_k|^{2}<\infty\Bigr\},
\qquad
\braket{f,g}=\sum_{k\ge0}a_k\overline{b_k},
\]
where $g=\sum_k b_kz^k$. The  \emph{unilateral shift} is the isometry $M_z:H^2\to H^2$, $(M_zf)(z)=zf(z)$; its adjoint is the \emph{backward shift} $M_z^\dagger\bigl(\sum_k a_kz^k\bigr)=\sum_k a_{k+1}z^k$.

An \emph{inner function} is a bounded analytic $f$ on $\oDisk$ with $|f|=1$ on $\partial\oDisk$.  The function $\varphi$ is the finite Blaschke product $\prod_j b_{\lambda_j}$. It determines the \emph{model space} $\Kf=H^2\ominus\varphi H^2$ (the orthogonal complement of $\varphi H^2$ in $H^2$ with $\dim \Kf = m$) and the \emph{compressed shift}
\[
S(\varphi)=P_{\Kf}\,M_z|_{\Kf},
\]
where $P_{\Kf}:H^2\to\Kf$ is the orthogonal projection. 

The \emph{Malmquist-Walsh basis} is an orthonormal basis of $\Kf$. For $j=1,\dots,m$, define
\[
e_j(z)=(-1)^{j}\,\frac{d_j}{1-\overline{\lambda_j}z}\,
\prod_{\ell>j}b_{\lambda_\ell}(z).
\]
In this basis the compressed shift $S(\varphi)=P_{\Kf}M_z|_{\Kf}$ has the upper-triangular matrix $S(\Lambda)$ of \cref{eq:livsic}.

Denote the defect space $\Dfr=\ran D_{A^\dagger }$. In finite dimensions $A^\dagger A$ and $AA^\dagger$ have the same nonzero eigenvalues, so
\[
d  =\dim\Dfr=\rk(I-A^\dagger A)=\rk(I-AA^\dagger).
\]

Let $H^2(\Dfr)=\{F(z)=\sum_{k\ge0}z^{k}\xi_k:\xi_k\in\Dfr,\ \sum\|\xi_k\|^{2}<\infty\}$ with $\|F\|^{2}=\sum_k\|\xi_k\|^{2}$. The shift $M_z$ acts by $(M_zF)(z)=zF(z)$, and its adjoint is the backward shift $M_z^\dagger \bigl(\sum_k z^k\xi_k\bigr)=\sum_k z^{k}\xi_{k+1}$. Choosing an orthonormal basis $f_1,\dots,f_{r_A}$ of $\Dfr$ gives $H^2(\Dfr)\cong\bigoplus^{r_A}H^2$ and $\Kf\otimes\Dfr\cong\bigoplus^{r_A}\Kf$, this explains the use of $r_A$ copies.

\subsubsection{Construction of the isometry}
Define the map
\[
W:\Hs_1\to H^2(\Dfr),\qquad
(Wx)(z)=D_{A^\dagger }(I-zA^\dagger )^{-1}x
=\sum_{k=0}^{\infty}z^{k}D_{A^\dagger }A^{\dagger k}x .
\,
\]
The $k$-th Hardy coefficient is $D_{A^\dagger }A^{\dagger k}x$.

\begin{lemma}
    $W$ is an isometry. 
\end{lemma}
\begin{proof} Since $D_{A^\dagger }^{2}=I-AA^\dagger $,
\[
\|D_{A^\dagger }A^{\dagger k}x\|^{2}
=\braket{{(I-AA^\dagger )A^{\dagger k}x},{A^{\dagger k}x}}
=\|A^{\dagger k}x\|^{2}-\|A^{\dagger(k+1)}x\|^{2},
\]
so the partial sums telescope
\[
\sum_{k=0}^{M}\|D_{A^\dagger }A^{\dagger k}x\|^{2}
=\|x\|^{2}-\|A^{\dagger(M+1)}x\|^{2}\xrightarrow[M\to\infty]{}\|x\|^{2}.
\]
To obtain the limit, note that since $A$ is completely non unitary, its eigenvalues are in $\oDisk$, thus $\|A^{\dagger(k)}\| \to 0$. Hence $\|Wx\|=\|x\|$, and $ W^\dagger  W=I_{\Hs_1}$.
\end{proof}

From the power series,
\[
WA^\dagger  x=\sum_{k\ge0}z^{k}D_{A^\dagger }A^{\dagger(k+1)}x
=M_z^\dagger  Wx ,
\]
so $WA^\dagger =M_z^\dagger  W$.

Taking adjoints, $AW^\dagger =W^\dagger  M_z$, hence $A^{k}W^\dagger =W^\dagger  M_z^{k}$; multiplying by $W$ and using $W^\dagger  W=I$,
\[
A^{k}=W^\dagger  M_z^{k}W\qquad(k\ge0).
\]
Thus, the shift $M_z$ is a power dilation of $A$. The relation $\varphi(A)=0$ collapses the range into a finite space. From $WA^\dagger =M_z^\dagger  W$, one gets $W\phi(A)^\dagger =\phi(M_z)^\dagger W =0$. Hence, $\ran W \subset \ker \phi(M_z)^\dagger = \Kf \otimes \Dfr $ and  $W = P_{\Kf}W$.

Pick orthonormal bases $x_1,\dots,x_q$ of $\Hs_1$, $f_1,\dots,f_{r_A}$ of $\Dfr$, and $e_1,\dots,e_m$ of $\Kf$. Then $\{e_j\otimes f_i\}$ is an orthonormal basis of $\Kf\otimes\Dfr$. Writing $e_j(z)=\sum_k c_k^{(j)}z^{k}$ and using $Wx_s=\sum_k z^{k}D_{A^\dagger }A^{\dagger k}x_s$,
\[
W_{(i,j),s}
=\sum_{k\ge0}\overline{c_k^{(j)}}\,\braket{{D_{A^\dagger }A^{\dagger k}x_s},{f_i}}
=\sum_{k\ge0}\overline{c_k^{(j)}}\,\braket{x_s,{A^{k}D_{A^\dagger }f_i}}
=\braket{{x_s},{e_j(A)D_{A^\dagger }f_i}}.
\]
\[
W_{(i,j),s}=\braket{{x_s},{e_j(A)D_{A^\dagger }f_i}} ,
\]
with
\[
e_j(A)=(-1)^{j}d_j(I-\overline{\lambda_j}A)^{-1}\prod_{\ell>j}b_{\lambda_\ell}(A),
\qquad
b_\lambda(A)=(\lambda I-A)(I-\overline{\lambda}A)^{-1}.
\]

\subsection{Global Isometry}

Both partial dilations can be assembled. Pad the sum that needs it up to $N=\max(r_U,r_A)$ copies by adjoining unused zero coordinates. This keeps $V_U$ and $W$ isometric, and the added coordinates lie outside their ranges. 
Regrouping the copies gives the coordinate permutation
\[
\Bigl(\bigoplus^{N}\C^{p}\Bigr)\oplus\Bigl(\bigoplus^{N}\Kf\Bigr)
\;\cong\;\bigoplus^{N}\bigl(\C^{p}\oplus\Kf\bigr),
\]
which carries
$\bigl(\bigoplus^{N}D\bigr)\oplus\bigl(\bigoplus^{N}S(\Lambda)\bigr)$ to
$B=\bigoplus^{N}T_1$, where $T_1=D\oplus S(\Lambda)$. Define
\[
J_0=V_U\oplus W:\ \Hs_0\oplus\Hs_1\ \longrightarrow\ \bigoplus^{N}(\C^{p}\oplus\Kf).
\]
Since $\Hs_0\perp\Hs_1$ and each summand is isometric, $J_0^\dagger J_0=I_{\Hs}$, so $J_0$ is an isometry. Moreover $V_U$ maps into the $\C^{p}$-block and $W$ into the $\Kf$-block, on which $B$ acts as $\bigoplus^{N}D$ and $\bigoplus^{N}S(\Lambda)$ respectively. 
Hence $J_0^\dagger B^{k}J_0$ splits along
$\Hs_0\oplus\Hs_1$:
\[
J_0^\dagger B^{k}J_0
=\Bigl(V_U^\dagger\bigl(\textstyle\bigoplus^{N}D\bigr)^{k}V_U\Bigr)
\oplus
\Bigl(W^\dagger\bigl(\textstyle\bigoplus^{N}S(\Lambda)\bigr)^{k}W\Bigr)
=U^{k}\oplus A^{k}=T^{k}\qquad(k\ge0).
\]
Finally, let $V \in U(m+1)$ and $J_1= \begin{psmallmatrix} \Id{m}\\ 0_{1\times m}\end{psmallmatrix}$ given by \cref{lem:defect_rank} such that $V$ is a dilation of $S(\Lambda)$, and $S(\Lambda) = J_1^\dagger V J_1$. 
Then $J = (\bigoplus^N (\Id{p} \oplus J_1))J_0 $ is the isometry corresponding to the dilation $U_{dil} = \bigoplus^N ( D \oplus V)$,
\[  
    T = J^\dagger U_{dil} J.
\]

\printbibliography

\end{document}